\documentclass[conference,a4paper]{IEEEtran}

\usepackage[utf8]{inputenc} 
\usepackage[T1]{fontenc}
\usepackage{url}              
\usepackage{cite}             

\usepackage[cmex10]{amsmath}  
\usepackage{mleftright}       
\mleftright                   

\usepackage{graphicx}         
\usepackage{booktabs}         
\usepackage{amsmath,amssymb,amsfonts,amsthm}
\usepackage{graphicx}
\usepackage{algorithm}
\usepackage{algorithmic}
\usepackage{amsthm}
\usepackage{enumerate}

\usepackage{xcolor} 

\begin{document}
	
	\title{Capacity-Achieving Sparse Regression Codes via Vector Approximate Message Passing} 
	
	\author{%
		\IEEEauthorblockN{Yizhou Xu\IEEEauthorrefmark{1},
			YuHao Liu\IEEEauthorrefmark{1},
			ShanSuo Liang\IEEEauthorrefmark{3},
			Tingyi Wu\IEEEauthorrefmark{3},
			Bo Bai\IEEEauthorrefmark{3},
			Jean Barbier,
			and TianQi Hou\IEEEauthorrefmark{3}}
		\IEEEauthorblockA{\IEEEauthorrefmark{1}%
			Department of Mathematical Sciences, Tsinghua University, Beijing, China}\
		\IEEEauthorblockA{\IEEEauthorrefmark{3}%
			Theory Lab, Central Research Institute, 2012 Labs, Huawei Technologies Co., Ltd.}
		\IEEEauthorblockA{%
			Emails: \{xu-yz19, yh-liu21\}@mails.tsinghua.edu.cn, 
			\{liang.shansuo, wu.ting.yi\}@huawei.com,
			ee.bobbai@gmail.com,\\
			jean.barbier.cs@gmail.com,
			thou@connect.ust.hk}
	}

	\maketitle
	
	\begin{abstract}
		Sparse regression codes (SPARCs) are a promising coding scheme that can approach the Shannon limit over Additive White Gaussian Noise (AWGN) channels. Previous works have proven the capacity-achieving property of SPARCs with Gaussian design matrices. We generalize these results to right orthogonally invariant ensembles that allow for more structured design matrices. With the Vector Approximate Message Passing (VAMP) decoder, we rigorously demonstrate the exponentially decaying error probability for design matrices that satisfy a certain criterion with the exponentially decaying power allocation. For other spectra, we design a new power allocation scheme to show that the information theoretical threshold is achievable.
	\end{abstract}
	
	\section{Introduction}
	Sparse regression codes (SPARCs), also called sparse superposition codes, are a computationally-efficient substitute for current coded modulation schemes that provably achieve the Shannon capacity over Additive White Gaussian Noise (AWGN) channels \cite{venkataramanan2019sparse}. Recent decades have witnessed significant advances in communication theory including in the areas of turbo codes\cite{berrou1996near}, LDPC codes\cite{gallager1962low} and polar codes\cite{arikan2009channel}; however, all the mentioned coding schemes are provably capacity-achieving only over discrete channels. Moreover, modern coded modulation schemes are also not guaranteed to be capacity-achieving over AWGN channels.
	
	SPARCs were originally proposed by Joseph and Barron in \cite{joseph2012least}, where it was shown that the error probability of the maximal likelihood decoder for SPARCs with Gaussian design matrices decays exponentially with respect to the code length $n$, for all communication rates $R$ smaller than the Shannon capacity $C$. As the computational complexity of the maximal likelihood decoder grows exponentially in $n$, subsequent to this work, two practical decoders with $O(n^2)$ complexity were proposed: the adaptive successive hard-decision decoder \cite{joseph2013fast} and the iterative soft-decision decoder \cite{barron2012high}. This paper focuses on a subset of iterative soft-decision decoders for SPARCs based on Approximate Message Passing (AMP) algorithms.
	
	AMP refers to a class of algorithms derived from belief propagation over dense factor graphs, originally studied in the context of compressed sensing. A key feature of AMP is the existence of a scalar recursion, referred to as State Evolution (SE), that can be used to accurately predict its performance under certain conditions. The AMP decoder for SPARCs with Gaussian design matrices is first given in \cite{barbier2017approximate}. Rush et al.\ gave the first fully rigorous proof that it achieves capacity over the AWGN channel in \cite{rush2017capacity} and demonstrated an exponentially decaying error probability in \cite{rush2018error}. Similar asymptotic performance results are obtained by Barbier et al.\ in \cite{barbier2017approximate,barbier2016proof,biyik2017generalized} via potential analysis.
	
	In practice, non-Gaussian designs can be used to reduce the computational complexity of the coding scheme. One approach is to consider a spatially coupled Gaussian design with the AMP decoder\cite{barbier2016proof, rush2021capacity}. With a computational complexity of $O(n \log n)$, both Hadamard and discrete cosine transform (DCT) matrices are also  useful in practice (see \cite{greig2017techniques}); moreover, they can be well approximated using orthogonally invariant matrices. For such orthogonally invariant matrices, however, AMP might deviate from its SE predictions \cite{rangan2019convergence}; thus, the analysis in \cite{rush2017capacity} showing SPARCs with AMP decoding achieves the AWGN capacity no longer holds.
	
	Therefore, a natural extension is to study AMP decoding for SPARCs with more structured design matrices, namely, right orthogonally invariant matrices. In this paper, we analyze a SPARCs decoder based on Vector Approximate Message Passing (VAMP), also known as Orthogonal Approximate Message Passing\cite{ma2016orthogonal} or Expectation Propagation\cite{takeuchi2017rigorous} introduced in \cite{rangan2019vector} for right orthogonally invariant ensembles. \cite{rangan2019vector} rigorously showed that the SE of VAMP almost surely characterizes its performance in the large system limit for compressed sensing problems, and later \cite{rush2022finite} showed that the convergence rate of the algorithm around its SE predictions is exponentially fast. \cite{hou2022sparse} physically demonstrated that SPARCs are capacity-achieving with right orthogonally invariant design matrices that satisfy a spectrum criterion, but a rigorous characterization of the decoder's performance is missing in the literature. 
	
	The main contribution of this work is to extend the non-asymptotic analysis in \cite{rush2022finite} to SPARCs with right orthogonally invariant design matrices and the VAMP decoder, giving the first fully rigorous analysis of this coding scheme. We show exponentially decaying error probability in Theorem \ref{theorem:main} for all $R<C$ (thus capacity-achieving) with spectra of design matrices satisfying a criterion (formally given in \eqref{eq:asymp_criterion}). We follow the framework in \cite{rush2022finite}, based on the conditioning technique introduced by \cite{bayati2011dynamics}, but we stress that the results in \cite{rangan2019vector} and \cite{rush2022finite} cannot be directly generalized to SPARCs, as the prior distribution on the signal is different and the sampling ratio in our settings tends to zero in the large system limit. For design matrices with arbitrary spectra, we also rigorously characterize the performance of the VAMP decoder in Theorem \ref{theorem:addition} through the algorithmic and information theoretical threshold introduced in \cite{hou2022sparse}.
	
	\section{SPARCs with the VAMP Decoder}
	In this work, we model SPARCs over AWGN channels as
	\begin{equation}
		\boldsymbol{y}= \boldsymbol{A}\boldsymbol{x}_0+\boldsymbol{w},
		\label{eq:model}
	\end{equation}
	where $\boldsymbol{w}\sim\mathcal{N}(\boldsymbol{0},\sigma^2\boldsymbol{I}_n)$ is the Gaussian noise, $\boldsymbol{y} \in \mathbb{R}^n$ is the channel output, $\boldsymbol{A} \in \mathbb{R}^{n \times N}$ is the design matrix, and $\boldsymbol{x}_0\in\mathbb{R}^N$ is the unknown message vector that satisfies an average power constraint $\frac{1}{n}\mathbb{E}[\|\boldsymbol{x}_{0}\|^2] = P$. The Shannon capacity of the AWGN channel is $C=\frac{1}{2}\log(1+\text{snr})$, where $\text{snr}=\frac{P}{\sigma^2}$, and the decoding task is to recover $\boldsymbol{x}_{0}$ from knowledge of $\boldsymbol{A}$ and $\boldsymbol{y}$. In SPARCs, $\boldsymbol{x}_0$ has $L$ sections with one non-zero element in each section forming a one-hot encoding of input bits, i.e.\ $
	\boldsymbol{x}_{0,j}=\sqrt{nP_\ell}$ or $0$ for $j\in \text{sec}(\ell)$ with $\ell=1,2,...,L,$
	where $\text{sec}(\ell):=\{j\in\mathbb{N}\;|\;M(\ell-1)<j\le M\ell\}$ and $N=ML$. $\boldsymbol{x}_0$ is uniformly distributed over all possible $M^L$ messages. The communication rate is therefore $R=\frac{1}{n}L\log M$, as each of the $M^L$ unique messages $\boldsymbol{x}_{0}$ corresponds in a one-to-one way with a $\log(M^L)$-length string of input bits.
	
	To enforce the average power constraint, we assume $\sum_{\ell=1}^{L}P_\ell=P$. The choice of $\{P_\ell\}_{\ell=1}^L$ is called the power allocation. A commonly-used power allocation in the theoretical analysis is the exponentially decaying power allocation
	$P_\ell=P\frac{e^{2C/L}-1}{1-e^{-2C}}e^{-2C\ell/L}$ for $\ell=1,2,...,L$.
	
	For asymptotic results, we assume $n$ and $N$ both approach infinity with $M=L^a$ for some $a>0$, while $R$ remains constant, and refer to this as the \emph{large system limit}, denoted as $\lim$. In this scheme, the sampling ratio is denoted as $\alpha :=\frac{n}{N}\to0$. For non-asymptotic results, we assume $L$ and $M$ are large enough, again with $M=L^a$.
	
	We consider a right orthogonally invariant coding matrix $\boldsymbol{A}=\boldsymbol{U}\boldsymbol{S}\boldsymbol{V}^T\in\mathbb{R}^{n\times N}$, where $\boldsymbol{U} \in\mathbb{R}^{n\times n}$ and $\boldsymbol{V} \in\mathbb{R}^{N\times N}$ are orthogonal matrices and $\boldsymbol{V}$ is Haar distributed. $\boldsymbol{S}=\text{Diag}(\boldsymbol{s}) \in\mathbb{R}^{n\times N}$ is a rectangular diagonal matrix with $\boldsymbol{s}\in\mathbb{R}^{n}$ composed of iid copies of the random variable $\sqrt{\frac{N}{n}}S$. We also require that $S\in[S_{min},S_{max}]$ is bounded ($S_{min}>0$) and $\mathbb{E}[S^2]=1$. It should be remarked that $S$ may implicitly depend on the system size. For example, its distribution is $\rho_S(s)=\sqrt{(s-s_-)(s_+-s)}/(2\pi s)$, where $s_\pm=(1\pm \sqrt{\alpha})^2$ for the Gaussian ensemble. Commonly used coding matrices, including Gaussian and row-orthogonal matrices ($\rho_S(s)=\delta(s-1)$) satisfy this condition.
	
	\begin{algorithm}[t]
		\caption{VAMP decoder}
		\begin{algorithmic}[1]
			\REQUIRE Max iteration T, right  orthogonally invariant matrix $\boldsymbol{A}\in\mathbb{R}^{n\times N}$, observation $\boldsymbol{y}\in\mathbb{R}^n$.
			\STATE Initialization: $\boldsymbol{r_{10}}=0$
			\FOR { t = 0 to T }   
			\STATE  $\boldsymbol{\hat{x}_{1t}}=g_1(\boldsymbol{r_{1t}},\bar{\gamma}_{1t})$; \quad $\boldsymbol{r_{2t}}=\frac{\boldsymbol{\hat{x}_{1t}}-\bar{\alpha}_{1t}\boldsymbol{r_{1t}}}{1-\bar{\alpha}_{1t}}$
			\STATE  $\boldsymbol{\hat{x}_{2t}}=g_2(\boldsymbol{r_{2t}},\bar{\gamma}_{2t})$; \quad $\boldsymbol{r_{1,t+1}}=\frac{\boldsymbol{\hat{x}_{2t}}-\bar{\alpha}_{2t}\boldsymbol{r_{2t}}}{1-\bar{\alpha}_{2t}}$
			\ENDFOR
			\RETURN $\boldsymbol{\hat{x}}$ obtained from $\boldsymbol{\hat{x}_{1T}}$ by setting the largest entry in section $\ell$ to $\sqrt{nP_\ell}$ and others to $0$
		\end{algorithmic}
		\label{algo:VAMP}
	\end{algorithm}
	
	We study the VAMP decoder given in Algorithm \ref{algo:VAMP}, where for $i \in \{1, 2, \ldots N\}$ and assuming $i \in \text{sec}(\ell)$,
	\begin{equation}
		[g_1(\boldsymbol{r},\bar{\gamma})]_i=\sqrt{nP_\ell}\frac{e^{\bar{\gamma}r_i\sqrt{nP_\ell}}}{\sum_{j\in\text{sec}(\ell)}e^{\bar{\gamma}r_j\sqrt{nP_\ell}}}
	\end{equation}
	is the Bayes optimal minimum mean-squared error (MMSE) estimator $\mathbb{E}[\boldsymbol{x}_0|\boldsymbol{x}_0+\mathcal{N}(0,\bar{\gamma}^{-1}\boldsymbol{I}_N)=\boldsymbol{r}]$, and
	\begin{equation}
		g_2(\boldsymbol{r},\bar{\gamma})=(\gamma_w\boldsymbol{A}^T\boldsymbol{A}+\bar{\gamma}\boldsymbol{I})^{-1}(\gamma_w\boldsymbol{A}^T\boldsymbol{y}+\bar{\gamma}\boldsymbol{r})
	\end{equation}
	with $\gamma_w=\frac{1}{\sigma^2}$ is a linear MMSE estimator from $\boldsymbol{y}\sim\mathcal{N}(\boldsymbol{A}\boldsymbol{x}_0,\bar{\gamma}^{-1}_w\boldsymbol{I})$. Here both denoisers are matched because they use the true prior and noise levels.
	
	The VAMP decoder in Algorithm \ref{algo:VAMP} is similar to the VAMP algorithm introduced in \cite{rangan2019vector}, with a key difference being that we directly use the SE prediction in the iteration, as in \cite{rush2017capacity} and \cite{rush2018error}. Recall that the SE is a set of recursions used to predict the performance, e.g.\ the mean square error (MSE), of the algorithm. The SE under the matched condition reads
	\begin{equation}
		\label{SE}
		\begin{aligned}
			&\bar{\alpha}_{1t}=\bar{\gamma}_{1t} \, \varepsilon_1(\bar{\gamma}_{1t}),\qquad\bar{\gamma}_{2t}=\frac{1}{\varepsilon_1(\bar{\gamma}_{1t})}-\bar{\gamma}_{1t},\\
			&\bar{\alpha}_{2t}=\bar{\gamma}_{2t} \, \varepsilon_2(\bar{\gamma}_{2t}),\qquad\bar{\gamma}_{1,t+1}=\frac{1}{\varepsilon_2(\bar{\gamma}_{2t})}-\bar{\gamma}_{2t},
		\end{aligned}
	\end{equation}
	initialized with $\bar{\gamma}_{1t}=P^{-1}$, where
	$\varepsilon_1(\bar{\gamma}_{1t})=\text{var}[X_0|R=X_0+\mathcal{N}(0,\bar{\gamma}_{1t}^{-1})]$ and $\varepsilon_2(\bar{\gamma}_{2t})=S^\alpha_{\gamma_w\boldsymbol{A}^T\boldsymbol{A}}(-\bar{\gamma}_{2t})$ with $S^\alpha_{\gamma_w\boldsymbol{A}^T\boldsymbol{A}}$ denoting the Stieltjes transform (see \cite{tulino2004random} for details) of $\gamma_w\boldsymbol{A}^T\boldsymbol{A}$, defined as 
	$S^\alpha_{\gamma_w\boldsymbol{A}^T\boldsymbol{A}}(\omega)=\frac{1}{N}\mathbb{E}[\text{Tr}[(\gamma_w\boldsymbol{A}^T\boldsymbol{A}-\omega\boldsymbol{I}_N)^{-1}]]$. The superscript here emphasizes that its value is related to the system size. Unlike common settings, the SE of SPARCs is naturally considered with finite sizes, as done in \cite{rush2018error} and \cite{rush2021capacity}.  
	
	\section{Finite Sample Analysis of the VAMP Decoder}
	Our goal in this section is to analyze the performance of the VAMP decoder for finite code lengths. Our analysis consists of two steps. We first show that the MSE, determined by the SE, approaches zero after a finite number of steps in Lemma \ref{lemma:SE}. Then we show that the VAMP decoder concentrates on its SE exponentially fast in Lemma \ref{lemma:convergence}, ensuring the SE performance predictions in Lemma \ref{lemma:SE} are accurate.
	
	\subsection{Decrease of the SE}
	We show in Lemma \ref{lemma:SE} that if the spectrum of the design matrix satisfies condition \eqref{eq:condition}, then the MSE given by $\varepsilon_1(\bar{\gamma}_{1T^*})$ will become arbitrarily small after $T^*$ iterations, for $M$ large enough. See \eqref{eq:convergence} and the surrounding discussion to see that  $\varepsilon_1(\bar{\gamma}_{1T^*})$ predicts the MSE performance.
	\newtheorem{lemma}{Lemma}
	\begin{lemma}
		\label{lemma:SE}
		We consider the exponentially decaying power allocation. Define $\Delta_R=\frac{C-R}{C}$, and without loss of generality, we suppose $\Delta_R<\frac{1}{2}$. Denote
		$$\Psi(z):=\lim\mathbb{E}\left[\frac{S}{1-zS}\right],$$ 
		where the expectation is with respect to the random variable $S$, which may implicitly depend on the system size. For a fixed $R$, if there exists a positive constant $c$ such that
		\begin{equation}
			\Psi(z)\geq\frac{1}{1-z+\frac{\Delta_R+\Delta_R^2}{2}(1-c)}
			\label{eq:condition}
		\end{equation}
		is satisfied for all $-\text{snr}\leq z\leq 0$, we have
		\begin{equation}	\varepsilon_1\left(\bar{\gamma}_{1T^*}\right)\leq P\alpha\left(1-f_R(M)\right),
		\end{equation}
		for $L$ and $M$ large enough, where where $f_R(M)=\frac{M^{-\kappa_1\Delta_R^2}}{\Delta_R\sqrt{\log M}}$ with $\kappa_1$ a universal constant and $T^*=\left\lceil\frac{2\text{snr}}{c(\Delta_R+\Delta_R^2)}\right\rceil$.
	\end{lemma}
	\begin{proof}
		Following \cite{rush2018error}, denote $\tau_t^2:={1}/{\bar{\gamma}_{1t}}$ and $x_t:=1-\frac{\varepsilon_1(\bar{\gamma}_{1t})}{\alpha P}$, where $\tau_t$ is interpreted as the effective noise level and $x_t$ as the weighted fraction of sections that have been correctly decoded after step $t$. Thus, from the definition of $\varepsilon_1(\bar{\gamma}_{1t})$ we have
		\begin{equation*}
			x_{t+1}=\sum_{\ell=1}^L\frac{P_\ell}{P} \mathbb{E}\left[\frac{e^{\frac{\sqrt{nP_\ell}}{\tau_{t+1}}\left(U_1^\ell+\frac{\sqrt{nP_\ell}}{\tau_{t+1}}\right)}}
			{e^{\frac{\sqrt{nP_\ell}}{\tau_{t+1}}\left(U_1^\ell+\frac{\sqrt{nP_\ell}}{\tau_{t+1}}\right)}+\sum_{j=2}^Me^{\frac{\sqrt{nP_\ell}}{\tau_{t+1}}U_j^\ell}}\right],
		\end{equation*}
		where $U_1,...,U_M\overset{iid}{\sim}\mathcal{N}(0,1)$.
		
		Using the lower bound on the right side in \cite[Lemma 2]{rush2018error},
		\begin{equation}
			x_{t+1}\geq1+(\sigma^2-\tau_{t+1}^2)/P+\chi\label{inequality},
		\end{equation}
		until $x_t$ reaches $1-f_R(M)$, where $\chi=\frac{\sigma^2}{P}\frac{\Delta_R+\Delta_R^2}{2}$. By the last equation in \eqref{SE},
		\begin{equation}
			\tau_{t+1}^2=\left[\frac{1}{S^\alpha_{\gamma_w\boldsymbol{A}^T\boldsymbol{A}}(-\bar{\gamma}_{2t})}-\bar{\gamma}_{2t}\right]^{-1}.
		\end{equation}
		Expanding the right side in the large system limit, we have
		\begin{equation}
			\tau_{t+1}^2=\frac{1}{\gamma_w\Psi(-\text{snr}(1-x_t))}+O(\alpha).\label{expansion}
		\end{equation}
		Therefore, combining \eqref{inequality} and \eqref{expansion}, once the condition \eqref{eq:condition} is satisfied, until $x_t\geq1-f_R(M)$, we have
		\begin{equation}
			x_{t+1}\geq x_t+c\chi+O(\alpha).
		\end{equation}
		Thus, $x_{t+1}\geq x_t+\frac{c}{2}\chi$ for sufficiently large $L$ and $M$. After $T^*=\left\lceil\frac{1}{c\chi}\right\rceil$ steps, $x_{T^*}\geq1-f_R(M)$, completing the proof.
	\end{proof}
	
	Recall from \cite[Lemma 1]{hou2022sparse} that $\Psi(z)\leq\frac{1}{1-z}$. Then the spectrum criterion proposed in \cite{hou2022sparse}, namely
	\begin{equation}
		\Psi(z)=\frac{1}{1-z},
		\label{eq:asymp_criterion}
	\end{equation}
	naturally follows from \eqref{eq:condition} if we require \eqref{eq:condition} to hold for all $R<C$.
	
	\subsection{Concentration on the SE}
	Having described the SE behavior, we now give the second ingredient in the proof of our main results, Lemma \ref{lemma:convergence}, which shows that the MSE of the VAMP decoder, $\frac{1}{2}||\boldsymbol{\hat{x}}_{1t}-\boldsymbol{x}_0||^2$, concentrates to its SE prediction, $\varepsilon_1(\gamma_{1t})$, exponentially fast.
	
	We follow \cite{rush2022finite} to show that VAMP concentrates on its SE, even if the prior of the signal is section-wise. First, we define 
	\begin{equation}
		\begin{aligned}
			\boldsymbol{p}_t:=\boldsymbol{r}_{1t}-\boldsymbol{x}_0, \qquad &\boldsymbol{q}_t:=\boldsymbol{V}^T(\boldsymbol{r}_{2t}-\boldsymbol{x}_0),\\ \boldsymbol{v}_t:=\boldsymbol{r}_{2t}-\boldsymbol{x}_0,  \qquad &\boldsymbol{u}_{t+1}:=\boldsymbol{V}^T(\boldsymbol{r}_{1,t+1}-\boldsymbol{x}_0),
		\end{aligned}
		\label{eq:define_gen}
	\end{equation}
	which follow a more general recursion:
	\begin{equation}
		\begin{aligned}
			&\boldsymbol{p}_t=\boldsymbol{V}\boldsymbol{u}_t,\quad\boldsymbol{v}_t=\frac{1}{1-\bar{\alpha}_{1t}}\left[f_p(\boldsymbol{p}_j,\boldsymbol{\omega}_p,\bar{\gamma}_{1t})-\bar{\alpha}_{1t}\boldsymbol{p}_t\right],\\
			&\boldsymbol{q}_t=\boldsymbol{V}^T\boldsymbol{v}_t,\quad\boldsymbol{u}_{t+1}=\frac{1}{1-\bar{\alpha}_{2t}}\left[f_q(\boldsymbol{q}_t,\boldsymbol{\omega}_q,\bar{\gamma}_{2t})-\bar{\alpha}_{2t}\boldsymbol{q}_t\right].
		\end{aligned}
		\label{General_recursion}
	\end{equation}
	Specifically, for the AWGN channel with $\xi\sim\mathcal{N}(0,1/\gamma_w)$,
	\begin{equation}
		\begin{split}
			f_p(p,\omega_p,\gamma)&=g_{1}(p+\omega_p,\gamma)-\omega_p,\label{fp},\\
			f_q(q,\omega_q,\gamma)&=\frac{\gamma_w\omega_q\xi+\gamma q}{\gamma_w\omega_q^2+\gamma}.
		\end{split}
	\end{equation}
	Here we denote $\boldsymbol{\omega}_p=\boldsymbol{x}_0$ and $\boldsymbol{\omega}_q=(\boldsymbol{s},\boldsymbol{0})\in\mathbb{R}^N$.
	
	In Lemma \ref{lemma:distribution},  we show that  the distributions of $\boldsymbol{p}_t$ and $\boldsymbol{q}_t$ can be described by a Gaussian part $\smash[t]{\overset{*}{\boldsymbol{p}}}_t$, $\smash[t]{\overset{*}{\boldsymbol{q}}}_t$ plus a deviation term, following the results of \cite{rush2022finite} and the statistical distributions are explicitly calculated via the conditioning technique. First we lay down
	some useful definitions and notation.
	
	For $t\geq0$, define matrices
	\begin{align*}
		\boldsymbol{U}_t:=  [\boldsymbol{u}_0|\ldots|\boldsymbol{u}_t],\qquad &\boldsymbol{V}_t:=  [\boldsymbol{v}_0|\ldots|\boldsymbol{v}_t]\in\mathbb{R}^{N\times(t+1)},\\
		\boldsymbol{P}_t:=  [\boldsymbol{p}_0|\ldots|\boldsymbol{p}_t],\qquad &\boldsymbol{Q}_t:=  [\boldsymbol{q}_0|\ldots|\boldsymbol{q}_t]\in\mathbb{R}^{N\times(t+1)},
	\end{align*}
	and for $t\geq1$,
	\begin{align*}
		\boldsymbol{C}_{pt}:=  [\boldsymbol{P}_t,\boldsymbol{V}_{t-1}],\qquad&\boldsymbol{C}_{ut}:=  [\boldsymbol{U}_t,\boldsymbol{Q}_{t-1}]\in\mathbb{R}^{N\times(2t+1)},\\
		\boldsymbol{C}_{vt}:=  [\boldsymbol{P}_{t-1},\boldsymbol{V}_{t-1}],\qquad&\boldsymbol{C}_{qt}:=  [\boldsymbol{U}_{t-1},\boldsymbol{Q}_{t-1}]\in\mathbb{R}^{N\times(2t)}.
	\end{align*}
	For $t\geq0$, define $\boldsymbol{B}_{\boldsymbol{C}_{pt}}^\perp\in\mathbb{R}^{N\times(N-2t-1)}$ to have columns that form an orthonormal basis for $\text{span}(\boldsymbol{C}_{pt})^\perp$, and $\boldsymbol{B}_{\boldsymbol{C}_{pt}}\in\mathbb{R}^{N\times(2t+1)}$ having columns that form an orthonormal basis for $\text{span}(\boldsymbol{C}_{pt})$ with $\boldsymbol{B}_{\boldsymbol{C}_{qt}}^\perp,\boldsymbol{B}_{\boldsymbol{C}_{ut}}^\perp,\boldsymbol{B}_{\boldsymbol{C}_{vt}}^\perp\in\mathbb{R}^{N\times(N-2t-1)}$ and $\boldsymbol{B}_{\boldsymbol{C}_{qt}},\boldsymbol{B}_{\boldsymbol{C}_{ut}},\boldsymbol{B}_{\boldsymbol{C}_{vt}}\in\mathbb{R}^{N\times(2t+1)}$ defined similarly.
	
	\begin{lemma}
		\label{lemma:distribution}
		For the general recursion in \eqref{General_recursion}, we have
		\begin{equation*}
			\begin{split}
				\boldsymbol{p}_t\overset{d}{=}\smash[t]{\overset{*}{\boldsymbol{p}}}_t+\sum_{r=0}^t[c_{pt}]_r\boldsymbol{\Delta}_{pr}, \qquad \boldsymbol{q}_t\overset{d}{=}\smash[t]{\overset{*}{\boldsymbol{q}}}_t+\sum_{r=0}^t[c_{qt}]_r\boldsymbol{\Delta}_{qr},
			\end{split}
		\end{equation*}
		where 
		\begin{equation}
			\begin{aligned}
				&([\smash[t]{\overset{*}{\boldsymbol{p}}}_0]_i,\ldots,[\smash[t]{\overset{*}{\boldsymbol{p}}}_t]_i)\overset{d}{=}(P_0,\ldots,P_t), \\
				&([\smash[t]{\overset{*}{\boldsymbol{q}}}_0]_i,\ldots,[\smash[t]{\overset{*}{\boldsymbol{q}}}_t]_i)\overset{d}{=}(Q_0,\ldots,Q_t),
			\end{aligned}
		\end{equation}
		with $(P_0,\ldots,P_t)$ being jointly Gaussian variables with 
		\begin{equation*}
			\boldsymbol{b}_{pt}:=(\mathbb{E}[P_0P_t],\ldots,\mathbb{E}[P_{t-1}P_t]),\quad\boldsymbol{\Sigma}_{pt}:=\text{Cov}(P_0,\ldots,P_t),
		\end{equation*}
		and $(Q_0,\ldots,Q_t)$ are jointly Gaussian variables with $\boldsymbol{b}_{qt}$ and $\boldsymbol{\Sigma}_{qt}$ defined similarly. The covariance matrices are calculated recursively by the SE of the general recursion, that is 
		\begin{align*}
			&\mathbb{E}[Q_jQ_t]=\\& \qquad \frac{\mathbb{E}[(f_p(P_j,W_p,\bar{\gamma}_{1j})-\bar{\alpha}_{1j}P_j) (f_p(P_t,W_p,\bar{\gamma}_{1t})-\bar{\alpha}_{1t}P_t)]}{\alpha}, \\
			&\mathbb{E}[P_{j+1}P_{t+1}]=\\
			&\qquad \frac{\mathbb{E}[(f_q(Q_j,W_q,\bar{\gamma}_{2j}) -\bar{\alpha}_{2j}Q_j)(f_q(Q_t,W_q,\bar{\gamma}_{2t})-\bar{\alpha}_{2t}Q_t)]}{(1-\bar{\alpha}_{2j})(1-\bar{\alpha}_{2t})},
		\end{align*}
		for $0\leq j\leq t$, initialized with $\boldsymbol{E}[P_0^2]=\bar{\gamma}_{10}^{-1}$. Further, the deviation terms are defined as
		\begin{align*}
			&\boldsymbol{\Delta}_{p0}=\left(\frac{||\boldsymbol{u}_0||}{||\boldsymbol{Z}_{p0}||}-\sqrt{\rho_{p0}}\right)\boldsymbol{B}_{\boldsymbol{C}_{v0}}\boldsymbol{Z}_{p0}, \\
			&\boldsymbol{\Delta}_{pt}=\boldsymbol{C}_{vt}\left((\boldsymbol{C}_{vt}^T\boldsymbol{C}_{vt})^{-1}\boldsymbol{C}_{qt}^T\boldsymbol{u}_t-\left[\begin{matrix}
				\boldsymbol{\beta}_{pt}\\\boldsymbol{0}
			\end{matrix}\right]\right)\\
			&\hspace{7mm} +\left[\frac{||[\boldsymbol{B}^\perp_{C_{qt}}]^T\boldsymbol{u}_t||}{||\boldsymbol{Z}_{pt}||}-\sqrt{\rho_{pt}}\right]\boldsymbol{B}^\perp_{\boldsymbol{C}_{vt}}\boldsymbol{Z}_{pt}-
			\sqrt{\rho_{pt}}\boldsymbol{B}_{\boldsymbol{C}_{vt}}\breve{\boldsymbol{Z}}_{pt}, \\
			&\boldsymbol{\Delta}_{q0}=\frac{\boldsymbol{p}_0^T\boldsymbol{v}_0}{||\boldsymbol{p}_0||^2}\boldsymbol{u}_0+\left[\frac{||[\boldsymbol{B}^\perp_{C_{p0}}]^T\boldsymbol{v}_0||}{||\boldsymbol{Z}_{q0}||}-\sqrt{\rho_{q0}}\right]\boldsymbol{B}^\perp_{\boldsymbol{C}_{u0}^\perp}\boldsymbol{Z}_{q0}\\
			&\hspace{7mm} -\sqrt{\rho_{q0}}\boldsymbol{B}^\perp_{\boldsymbol{C}_{u0}}\breve{\boldsymbol{Z}}_{q0}, \\
			&\boldsymbol{\Delta}_{qt}=\boldsymbol{C}_{ut}\left((\boldsymbol{C}_{pt}^T\boldsymbol{C}_{pt})^{-1}\boldsymbol{C}_{pt}^T\boldsymbol{v}_t-\left[\begin{matrix}
				\boldsymbol{0}\\\boldsymbol{\beta}_{pt}
			\end{matrix}\right]\right)\\
			&\hspace{7mm} +\left[\frac{||[\boldsymbol{B}^\perp_{C_{pt}}]^T\boldsymbol{v}_t||}{||\boldsymbol{Z}_{qt}||}-\sqrt{\rho_{qt}}\right]\boldsymbol{B}^\perp_{\boldsymbol{C}_{ut}}\boldsymbol{Z}_{qt}-\sqrt{\rho_{qt}}\boldsymbol{B}_{\boldsymbol{C}_{ut}}\breve{\boldsymbol{Z}}_{qt},
		\end{align*}
		where $\bar{\boldsymbol{Z}}_{pt}=[\boldsymbol{Z}_{pt}|\breve{\boldsymbol{Z}}_{pt}]\sim\mathcal{N}(0,\boldsymbol{I}_N)$ independently of $\bar{\boldsymbol{Z}}_{qt}=[\boldsymbol{Z}_{qt}|\breve{\boldsymbol{Z}}_{qt}]\sim\mathcal{N}(0,\boldsymbol{I}_N)$. $\boldsymbol{\beta}_{pt}:=[\boldsymbol{\Sigma}_{u(t-1)}]^{-1}\boldsymbol{b}_{ut}$ and $\rho_{pt}:=\mathbb{E}[P_t^2]-\boldsymbol{b}_{ut}^T[\boldsymbol{\Sigma}_{p(t-1)}]^{-1}\boldsymbol{b}_{ut}$, with $\boldsymbol{\beta}_{qt}$, $\rho_{qt}$ defined similarly. The coefficients of deviation terms are defined recursively with $[c_{pt}]_t=[c_{qt}]_t=1$ as
		\begin{equation*}
			[c_{pt}]_r=\sum_{i=r}^{t-1}[c_{pi}]_r \,\ [\boldsymbol{\beta}_{pt}]_{i+1},\qquad [c_{qt}]_r=\sum_{i=r}^{t-1}[c_{qi}]_r \, [\boldsymbol{\beta}_{qt}]_{i+1}.
		\end{equation*}

	\end{lemma}
	The proof of Lemma \ref{lemma:distribution} follows from \cite[Lemma 3 and 4]{rush2022finite}.
	Recall that SE is exact for Gaussian variables, i.e.\ it is exact if we were to neglect the deviation terms. The main technical aspect of the proof is showing that the magnitude of the deviation terms concentrate on zero. 
	
	\begin{lemma}
		\label{lemma:convergence}
		For $t\geq0$, let
		\begin{equation}
			\begin{aligned}
				&K_t=K^{2t}(t!)^{12},\ K_t'=\frac{K_t}{K(t+1)^6},\\ &\kappa_t=\frac{1}{\kappa^{2t}(t!)^{17}},\ \kappa_t'=\frac{\kappa_t}{\kappa(t+1)^8},
			\end{aligned}
		\label{eq:constants}
		\end{equation}
		where $K$ and $\kappa$ denote constants that do not depend on the iteration number $t$.
		Given a power allocation satisfying $P_\ell=\Theta(\frac{1}{L})$ and $T$ not related to $M$, for $L$ and $M$ large enough, the following results hold for $0\leq t\leq T$.
		
		\textbf{(a)} For all $0\leq j\leq t$,
		\begin{equation*}
			\begin{aligned}
				&\mathbb{P}\left(\left|\frac{1}{N}\boldsymbol{p}_j^T\boldsymbol{p}_t-\left[\boldsymbol{\Sigma}_{pt}\right]_{j+1,t+1}\right|\geq\epsilon\right)\\
				&\qquad\leq (t+1)^2KK_t'\exp\left\{-\frac{\kappa\kappa_t'L\epsilon^2}{(t+1)^4(\log M)^{2t+2}}\right\}
			\end{aligned}
		\end{equation*}
		\begin{equation*}
			\begin{aligned}
			&\mathbb{P}\left(\left|\frac{1}{n}\boldsymbol{q}_j^T\boldsymbol{q}_t-\frac{[\boldsymbol{\Sigma}_{qt}]_{j+1,t+1}}{\alpha} \right|\geq\epsilon\right)\\
			&\qquad\leq (t+1)^2KK_t\exp\left\{-\frac{\kappa\kappa_tL\epsilon^2}{(t+1)^4(\log M)^{2t+2}}\right\}
			\end{aligned}
		\end{equation*}
		\textbf{(b)} For all $0\leq i\leq t$ and $1\leq j\leq t+1$
		\begin{equation*}
			\begin{aligned}
				&\mathbb{P}\left(\left|\frac{1}{n}\boldsymbol{v}_i^T\boldsymbol{v}_t-[\boldsymbol{\Sigma}_{qt}]_{i+1,t+1}/\alpha\right|\geq\epsilon\right)\\
				&\qquad\qquad\leq (t+1)^2KK_t'\exp\left\{-\frac{\kappa\kappa_t'L\epsilon^2}{(t+1)^5(\log M)^{2t+2}}\right\},
			\end{aligned}
		\end{equation*}
		\begin{equation*}
			\begin{aligned}
			&\mathbb{P}\left(\left|\frac{1}{N}\boldsymbol{u}_j^T\boldsymbol{u}_{t+1}-[\boldsymbol{\Sigma}_{pt}]_{j+1,t+2}\right|\geq\epsilon\right)\\
			&\qquad\qquad\leq (t+1)^2KK_t\exp\left\{-\frac{\kappa\kappa_tL\epsilon^2}{(t+1)^4(\log M)^{2t+2}}\right\}.
		\end{aligned}
		\end{equation*}
	\end{lemma}
	
	The proof of Lemma \ref{lemma:convergence} is given in Appendix \ref{app:finite sample}. Its idea is similar to \cite[Lemma 5]{rush2022finite}, except that we need to specially deal with the non-Lipschitz estimation functions in \eqref{fp}. The treatment of the former one is borrowed from \cite{rush2018error}, while we use the concentration of sub-exponential variables to deal with concentration of the latter one.
	
	The $\log M$ factor in the denominator comes from the fact that the largest deviation among $M$ entries is approximately $\sqrt{\log M}$ times larger than the deviation of each entry.
	
	As the SE of the general recursion and of VAMP coincide (both depend on the system size), for $0\leq t\leq T^*$, a natural corollary of Lemma \ref{lemma:convergence}(b) is 
	\begin{equation}
		\begin{split}
			&P\left(\left|\frac{||\boldsymbol{\hat{x}}_{1t}-\boldsymbol{x}_0||^2 }{n}-\varepsilon_1(\bar{\gamma}_{1t})\right|\geq\epsilon\right)\\
			&\qquad\qquad\leq (t+1)^2KK_t'\exp\left\{-\frac{\kappa\kappa_t'L\epsilon^2}{(t+1)^5(\log M)^{2t+2}}\right\}. \label{eq:convergence}
		\end{split}
	\end{equation}
	 Indeed, from \eqref{eq:define_gen}-\eqref{fp} along with Algorithm~\ref{algo:VAMP}, one can easily see that $\boldsymbol{\hat{x}}_{1t}-\boldsymbol{x}_0 = f_p(\boldsymbol{p}_t, \boldsymbol{\omega}_p, \bar{\gamma}_{1t})$, so one can obtain \eqref{eq:convergence} from Lemma \ref{lemma:convergence}.

	\section{Performance of the VAMP Decoder}
	\subsection{Error Exponent for Design Matrices Satisfying the Spectrum Criterion}
	Here we demonstrate that the section error rate of SPARCs with the VAMP decoder and the spectrum criterion decays exponentially with section size $L$ for $R<C$, and thus is asymptotically capacity-achieving. It is formally stated in the following theorem.
	\newtheorem{theorem}{Theorem}
	\begin{theorem}
		\label{theorem:main}
		Let $ \epsilon_{\text{sec}}:=\frac{1}{L}\sum_{\ell=1}^L\mathbf{1}\{\hat{\boldsymbol{x}}_{\sec(\ell)}\neq\boldsymbol{x}_{0,\sec(\ell)}\}$ be the section error rate and $\kappa, K$ be universal constants.
		Under the Lemma \ref{lemma:SE}  condition, for $L, M$ sufficiently large 
		\begin{equation}
			\begin{aligned}
				&\mathbb{P}\left(\epsilon_{\text{sec}}>\epsilon\right)\\
				&\qquad\leq K_{T^*+1}'\exp\left\{\frac{-\kappa_{T^*+1}' L}{(\log M)^{2(T^*+1)}}\left[\frac{\epsilon\sigma^2C}{2}-Pf_R(M)\right]^2\right\},
			\end{aligned}
		\end{equation}
		where constants $K_{T^*+1}'$ and $\kappa_{T^*+1}'$ are defined in \eqref{eq:constants} with $T^*$ defined in Lemma \ref{lemma:SE}.
	\end{theorem}
	Lemma \ref{lemma:SE} gives the upper bound of the MSE predicted by the SE, and Equation \eqref{eq:convergence} shows the concentration of the VAMP decoder on its SE.
	As the section error rate can be controlled by the MSE, the proof of Theorem~\ref{theorem:main} follows analogously to the proof of \cite[Theorem 1]{rush2018error}.
	
	There are several things to notice with regard to Theorem \ref{theorem:main}: \textbf{(1)} The probability measure is over the right orthogonally invariant
	design matrix $\boldsymbol{A}$, the Gaussian noise $\boldsymbol{w}$ and uniformly distributed message $\boldsymbol{x}_0$. \textbf{(2)} Theorem \ref{theorem:main} also implies an asymptotic result $\lim \epsilon_{\text{sec}}=0$ for a fixed $R<C$. \textbf{(3)} Theorem \ref{theorem:main} specifies the performance trade-offs between $M$ and $L$. Larger $M$ yields smaller MSE in SE, but a larger probability that VAMP deviates from its SE. A numerical demonstration is presented in Appendix \ref{app:numerical}. \textbf{(4)} As an important corollary, once the spectrum criterion \eqref{eq:asymp_criterion} is satisfied and all singular values are strictly positive, SPARCs are capacity-achieving with an exponentially decaying error rate under the VAMP decoder. Both Gaussian and row-orthogonal matrices are in this case. \textbf{(5)} We can also consider that $R$ increases with $M$. As Lemma \ref{lemma:SE} requires $\chi$ to be strictly positive, we need $\lim\frac{f_R(M)}{\Delta_R}=0$. The permitted gap between $R$ and $C$ is then given by $ \Delta_R\geq\sqrt{\frac{\log\log M}{\kappa_1\log M}}$.
	
	\subsection{Asymptotic for Design Matrices with Arbitrary Spectra}
	For matrices not satisfying the spectrum criterion, we can also rigorously describe the performance of SPARCs by proving the algorithmic and information theoretical thresholds given in \cite{hou2022sparse}. More specifically, we give the necessary and sufficient condition (algorithmic threshold $R_{\text{alg}}$) for asymptotically error-free decoding with the average power allocation, and design a new power allocation that can asymptotically reach zero error rate for communication rates below the information theoretical threshold $R_{\text{IT}}$.
	\begin{theorem}
		For arbitrary spectra, when the average power allocation $P_\ell=\frac{P}{L}$ is used, we have $\lim\epsilon_{\text{sec}}=0$ if and only if $R<R_{\text{alg}} := \frac{1}{2} \, \text{snr} \, \Psi(-\text{snr})$. Denote 
		\begin{equation}
			F(x)=\frac{1}{\text{snr} \, \Psi(-\text{snr}\,(1-x))},
		\end{equation}
		and assume that it is continuously differentiable in $[0,1]$. If we design the power allocation as $P_\ell=P\frac{R_{\text{IT}}}{L}G(\frac{\ell}{L})$ with $R_{\text{IT}}:=\frac{1}{2}\int_0^{\text{snr}}\Psi(-u)du$ and
		\begin{equation}
			G^{-1}(2t)=\int_{F(0)}^t\frac{[F^{-1}(\xi)]'}{2R_{\text{IT}}\xi} \, d\xi,
			\label{newPA}
		\end{equation}
		then  $\lim\epsilon_{\text{sec}}=0$ for all $R<R_{\text{IT}}$ and if the spectrum criterion \eqref{eq:asymp_criterion} is satisfied, $R_{\text{IT}}=C$.
		\label{theorem:addition}
	\end{theorem}
	We note that according to its definition, $F(x)$ is positive and invertible for $x\in[0,1]$.
	\begin{proof}
		As the power allocation given in this theorem satisfies the condition in Lemma \ref{lemma:convergence}, we only need to show that after a finite number steps, the MSE of the SE recursion becomes asymptotically zero, or equivalently, $x_t$ tends to one.
		
		Let $c_\ell=LP_\ell$. Combining \eqref{expansion} and 
		\cite[Lemma 1]{rush2017capacity}, the asymptotic SE equation reads
		\begin{equation}
			\begin{aligned}
				&\bar{x}_t=\lim\sum_{\ell=1}^L \frac{P_{\ell}}{P} \mathbf{1}\left\{c_\ell>2 R \tau_t^2\right\},\\&
				\bar{\tau}_{t+1}^2 = \frac{1}{\text{snr} \, \Psi (- \text{snr} (1 - \bar{x}_t))},\label{asymptoticSE}
			\end{aligned}
		\end{equation} 
		where $\bar{x}_t=\lim x_t$ and $ \bar{\tau}_t=\lim\tau_t$. Recall that the initialization is $\lim x_0=0$. Thus, for average power allocation,
		\begin{equation}
			\bar{x}_2 = \mathbf{1} \left\{ R < \, \frac{1}{2}\text{snr} \, \Psi(-\text{snr})\right\}.
		\end{equation}
		Therefore, $\bar{x}_1=1$ (successful decoding) if $R<R_{\text{alg}}$ and zero otherwise.
		
		For the power allocation $P\frac{R_{\text{IT}}}{L}G(\frac{\ell}{L})$, first let us verify that it is normalized. From the lower limit of the integral in \eqref{newPA} we have $G^{-1}(2F(0))=0$, or $F^{-1}(\frac{1}{2}G(0))=0$. Using \eqref{newPA}
		\begin{equation}
			R_{\text{IT}}=\int_0^1\frac{1}{2F(x)}dx=\int_{F(0)}^{F(1)}\frac{1}{2t}[F^{-1}(t)]'dt,
		\end{equation}
		and we have 
		\begin{equation}
			\int_{F(0)}^{F(1)}[G^{-1}(2t)]'dt=1.
		\end{equation}
		Therefore, $G^{-1}(2F(1))=1$, or $F^{-1}(\frac{1}{2}G(1))=1$. 
		
		Finally, the power allocation is asymptotically normalized:
		\begin{align*}
			R_{\text{IT}}\int_0^1G(x)dx&=\int_{\frac{1}{2}G(0)}^{\frac{1}{2}G(0)}[F^{-1}(t)]'dt\\
			&=F^{-1}\left(\frac{1}{2}G(1)\right)-F^{-1}\left(\frac{1}{2}G(0)\right)=1.
		\end{align*}
		
		Next, given the power allocation \eqref{newPA}, the asymptotic SE in \eqref{asymptoticSE} becomes
		\begin{equation}
			\bar{x}_{t}=\min\left\{F^{-1}\left(\frac{R}{R_{\text{IT}}}\, \bar{\tau}_t^2\right), \, 1 \right\},\qquad\bar{\tau}_{t+1}^2=F(\bar{x}_t),
		\end{equation}
		which leads to $\bar{\tau}_{t+1}^2=\frac{R}{R_{\text{IT}}}\bar{\tau}_t^2$ until $\bar{\tau}_T^2=F(1)>0$ with $T=\lceil\frac{\log F(0)-\log F(1)}{\log R-\log R_{\text{IT}}}\rceil$. Therefore, we have $\lim x_T=1$ after finite steps.
	\end{proof}
	For matrices satisfying the spectrum criterion, we have $F(x)=\frac{1}{\text{snr}}+1-x$ and $G(x)=\frac{2(1+\text{snr})}{\text{snr}}e^{-2Cx}$, recovering the exponentially decaying power allocation as a special case.
	
	\section{Future Work}
	
	We expect that these results also hold for general channels, as pointed out non-rigorously by \cite{liu2022sparse} from the perspective of statistical physics. There are many open questions related to SPARCs and AMP-style decoding, including analyzing the performance of SPARCs with more general design matrices, like the class of semi-random matrices\cite{dudeja2022universality, dudeja2022spectral} and spatially coupled right orthogonally invariant matrices\cite{takeuchi2022long}. These ensembles are very important in practice to obtain lower computational complexity and better performance. Moreover, while SPARCs and AMP decoding have been studied in various modern channel models, like unsourced random access \cite{fengler2019sparcs, fengler2020unsourced, amalladinne2020approximate, amalladinne2021unsourced} or many-user multiple access models \cite{hsieh2022near, hsieh2021spatially, hsieh2021near}, no work has rigorously studied AMP-style decoders with more general design matrices in such settings. Finally we mention that recent work  \cite{li2022non} provides a finite sample analysis of an AMP algorithm for low rank matrix estimation and demonstrates concentration up to $O(\frac{n}{\text{poly} \log n})$ iterations and an interesting open question is whether proof techniques similar to that used in \cite{li2022non} could be used to analysis the AMP decoder for SPARCs to find improved error rates.
	\section{acknowledgements}
	We thank Cynthia Rush for helpful discussions.
	
	\bibliographystyle{unsrt}
	\bibliography{SRC_VAMP}
	
	
	\appendix
	\subsection{Numerical Results}\label{app:numerical}
	\begin{figure}[tb]
		\centering	
		\includegraphics[width=1.05\linewidth]{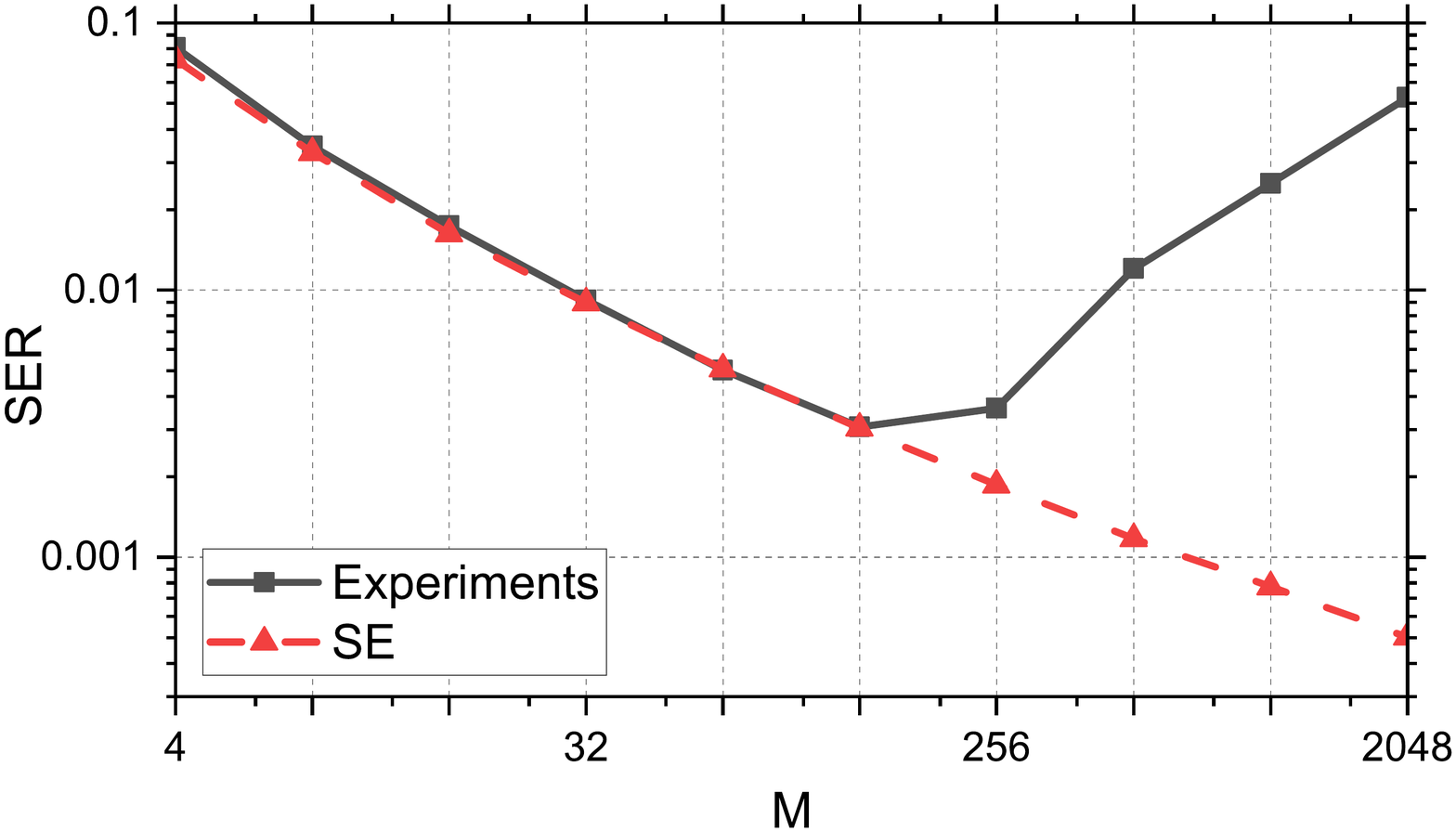}
		\caption{VAMP error performance with increasing $M$, for fixed $L = 1024$, $R = 1.5$ bits, and $\text{snr} = 11.1$ (2 dB from Shannon limit). The solid line is the empirical section error rate (SER), obtained by averaging over 200 trials for each point. The dashed line is the SER predicted by the SE. The design matrix is chosen to be a DCT matrix, and the iterative power allocation proposed in  \cite{greig2017techniques} is used as it provides better empirical performance than the exponentially decaying power allocation.
			\label{fig:fig1}}
	\end{figure}
	
	We first show the effect of increasing $M$ with fixed $L$ in Fig.\ \ref{fig:fig1}. When $M$ increases, the SE prediction of the error rate keeps decreasing; however, the experimental SER only decreases up to $M=128$, and then starts to increase. This phenomenon can be described by Theorem \ref{theorem:main}, because for larger $M$, the SE prediction $f_R(M)$ is smaller, but there will be a larger probability that the empirical performance diverges from the SE prediction due to the $\log M$ factor. This phenomenon was also reported for the AMP decoder of SPARCs \cite{rush2018error,greig2017techniques}.
	
	\subsection{Proof of Lemma \ref{lemma:convergence}}\label{app:finite sample}
	First, notice that in the Bayes optimal setting, the covariance matrices in Lemma \ref{lemma:distribution} can be calculated explicitly.
	\begin{lemma}
		\label{lemma:bound}
		When the non-linear functions take the forms in \eqref{fp}, with $g(r,\gamma)=\mathbb{E}[X_0|X_0+\mathcal{N}(0,\gamma^{-1})=r]$ being the posterior estimation, for $0\leq j\leq t$, we have 
		\begin{equation}
			\mathbb{E}[P_jP_t]=\mathbb{E}[P_t^2]=\bar{\gamma}_{1t}^{-1},\label{Pt^2}
		\end{equation}
		\begin{equation}
			\mathbb{E}[Q_jQ_t]=\mathbb{E}[Q_t^2]=\bar{\gamma}_{2t}^{-1}\label{Qt^2}.
		\end{equation}
		\label{lemma:covariance}
	\end{lemma}
	\begin{proof}
		Our proof is inductive. From its definition, $\mathbb{E}[P_0^2]=\bar{\gamma}_{1t}^{-1}$. 
		Assume that \eqref{Pt^2} holds for $0,\ldots,t$, then
		\begin{equation}
			\mathbb{E}[Q_t^2]=\frac{\varepsilon_1(\bar{\gamma}_{1t})-\bar{\alpha}_{1t}^2/\bar{\gamma}_{1t}}{(1-\bar{\alpha}_{1t})^2}=\bar{\gamma}_{2t}^{-1}.
		\end{equation}
		The second equality is due to the fact that $\bar{\alpha}_{1t}=\bar{\gamma}_{1t}\varepsilon_1(\bar{\gamma}_{1t})$. 
		
		We then notice that
		\begin{equation}
			\mathbb{E}[X_0|X_0+P_j,X_0+P_t]=\mathbb{E}[X_0|X_0+P_t],
		\end{equation}
		as $\mathbb{E}[P_t(P_t-P_j)]=0$. This will lead to
		\begin{equation}
			\mathbb{E}[f_p(P_j,W_p,\bar{\gamma}_{1j})f_p(P_t,W_p,\bar{\gamma}_{1t})]=\mathbb{E}[f_p(P_t,W_p,\bar{\gamma}_{1t})^2],\label{43}
		\end{equation}
		with
		\begin{equation}
			\mathbb{E}[f_p(P_t,W_p,\bar{\gamma}_{1t})^2]=\mathbb{E}[P_tf_p(P_t,W_p,\bar{\gamma}_{1t})]=\bar{\alpha}_{1t}\bar{\gamma}_{1t}^{-1}.\label{44}
		\end{equation}
		The first equality in \eqref{44} is due to the orthogonality principle, and the second equality is due to Stein's lemma.
		
		Using \eqref{43} and \eqref{44} and Stein's lemma in $\mathbb{E}[Q_jQ_t]=\mathbb{E}[(f_p(P_j,W_p,\bar{\gamma}_{1j})-\bar{\alpha}_{1t}P_j)[(f_p(P_t,W_p,\bar{\gamma}_{1j})-\bar{\alpha}_{1t}P_t)]$, we can obtain $\mathbb{E}[Q_jQ_t]=\mathbb{E}[Q_t^2]$, giving the result \eqref{Qt^2} for $t$.
		
		Next, we assume that \eqref{Qt^2} holds for $0,\ldots,t$. Then
		\begin{equation}
			\mathbb{E}[P_{t+1}^2]=\frac{\varepsilon_2(\bar{\gamma}_{2t})-\bar{\alpha}_{2t}^2/\bar{\gamma}_{1t}}{(1-\bar{\alpha}_{2t})^2}=\bar{\gamma}_{1,t+1}^{-1},
		\end{equation}
		where the second equality is due to the fact that $\bar{\alpha}_{2t}=\bar{\gamma}_{2t}\varepsilon_2(\bar{\gamma}_{2t})$. We further have
		\begin{equation}
			\begin{aligned}
				&\mathbb{E}[f_q(Q_j,W_q,\bar{\gamma}_{1j})f_q(Q_t,W_q,\bar{\gamma}_{2t})]\\
				&=\mathbb{E}\left[\frac{\gamma_wW_q^2+\bar{\gamma}_{2j}\bar{\gamma}_{2t} Q_jQ_t)}{(\gamma_wW_q^2+\bar{\gamma}_{2j})(\gamma_wW_q^2+\bar{\gamma}_{2t})}\right]\\
				&=\mathbb{E}\left[\frac{1}{\gamma_wW_q^2+\bar{\gamma}_{2t}} \right]=\mathbb{E}\left[f_q(Q_t,W_q,\bar{\gamma}_{2t})^2\right],
			\end{aligned}
		\end{equation}
		where the second equality is from \eqref{Qt^2}. Similarly, we then have $\mathbb{E}[P_{j+1}P_{t+1}]=\mathbb{E}[P_{t+1}^2]$. This completes the proof.
	\end{proof}
	
	Lemma \ref{lemma:covariance} is useful to obtain boundedness of certain variables in the following lemma.
	
	\begin{lemma}
		Under the conditions in Theorem \ref{theorem:main} or \ref{theorem:addition}, for $0\leq t\leq T^*$, we have $\bar{\gamma}_{1t},\alpha\bar{\gamma}_{2t}>\gamma_{min}$,  $\alpha_{min}<\frac{1}{\alpha}\bar{\alpha}_{1t},\frac{1}{\alpha}(1-\bar{\alpha}_{2t})<\alpha_{max}$, and $\rho_{pt},\rho_{qt}>\rho_{min}>0$, where $\gamma_{min}$, $\alpha_{min}$, $\alpha_{max}$ and $\rho_{min}$ are not related to $L$ and $M$.
	\end{lemma}
	\begin{proof}
		Under the conditions in Theorem \ref{theorem:main} or \ref{theorem:addition}, we both have
		$\bar{\gamma}_{10}^{-1}>\ldots>\bar{\gamma}_{1T^*}^{-1}>\sigma^2,$
		bounded. Thus, $\bar{\gamma}_{1t}>\bar{\gamma}_{10}$ and
		$P\bar{\gamma}_{10}<\frac{1}{\alpha}\bar{\alpha}_{1t}=\bar{\gamma}_{1t}\varepsilon_1(\bar{\gamma}_{1t})/\alpha<P\sigma^{-2}.$
		Correspondingly,
		\begin{equation*}
			\begin{aligned}
				\bar{\alpha}_{2t}=\bar{\gamma}_{2t}\varepsilon_2(\bar{\gamma}_{2t})&=\bar{\gamma}_{2t}S^\alpha_{\gamma_w\boldsymbol{A}^T\boldsymbol{A}}(-\bar{\gamma}_{2t})\\
				&=(1-\alpha)+\alpha\mathbb{E}\left[\frac{\bar{\gamma}_{2t}}{\gamma_wS^2/\alpha+\bar{\gamma}_{2t}}\right]
			\end{aligned}
		\end{equation*}
		where we recall that $S^\alpha_{\gamma_w\boldsymbol{A}^T\boldsymbol{A}}$ is the Stieltjes transform of $\gamma_w\boldsymbol{A}^T\boldsymbol{A}$.
		From Lemma \ref{lemma:SE}, we have $\alpha\bar{\gamma}_{2t}=\frac{1}{1-x_t}-\frac{\alpha}{\tau_t^2}>\gamma_{min}$. Thus, the following is both lower and upper bounded
		\begin{equation}
			\frac{1}{\alpha}(1-\bar{\alpha}_{2t})=1+\mathbb{E}\left[\frac{\alpha\bar{\gamma}_{2t}}{\gamma_wS^2+\alpha\bar{\gamma}_{2t}}\right].
		\end{equation}
		
		Lastly, from Lemma \ref{lemma:covariance} we can obtain
		\begin{equation*}
			\rho_{pt}:=\mathbb{E}[P_t^2]-\boldsymbol{b}_{ut}^T[\boldsymbol{\Sigma}_{p(t-1)}]^{-1}\boldsymbol{b}_{ut}=\bar{\gamma}_{1t}^{-1}\left(1-\frac{\bar{\gamma}_{1(t-1)}^{-1}}{\bar{\gamma}_{1t}^{-1}}\right),
		\end{equation*}
		where the second equality uses $[\boldsymbol{\Sigma}_{p(t-1)}]^{-1}(1,1,\ldots,1)^T=(0,\ldots,0,\bar{\gamma}_{1(t-1)}^{-1})$. Following \cite[Lemma 3]{rush2018error}, we can show that $\rho_{pt}$ is lower bounded by a positive constant. Similarly, $\alpha\rho_{qt}$ is also lower bounded by a positive constant.
	\end{proof}
	
	Following \cite{rush2022finite}, for $t\geq0$ with $\boldsymbol{\Pi}_{p0}=\boldsymbol{\Sigma}_{p0}^{-1}$, we define
	\begin{equation}
		\begin{aligned}
			\boldsymbol{\Pi}_{pt}=\left[\begin{matrix}
				\boldsymbol{\Sigma}_{pt}^{-1}&0\\0&\boldsymbol{\Sigma}_{q(t-1)}^{-1}
			\end{matrix}\right],\quad
			\boldsymbol{\Pi}_{qt}=\left[\begin{matrix}
				\boldsymbol{\Sigma}_{pt}^{-1}&0\\0&\boldsymbol{\Sigma}_{qt}^{-1}
			\end{matrix}\right].
		\end{aligned}
	\end{equation}
	
	Now we are ready to prove Lemma \ref{lemma:convergence}, given that part (a) and (b) of Lemma \ref{lemma:convergence} are also part (a) and (b) of Lemma \ref{lemma:main}.
	\begin{lemma}
		\label{lemma:main}
		The values of $K_t,K_t',\kappa_t,\kappa_t'$ for $t\geq0$ are given in \eqref{eq:constants}. Let $X_n\overset{.}{=}c$ denote
		\begin{equation}
			\mathbb{P}\left(\left|X_n-c \right|\leq\epsilon\right)\leq (t+1)^2KK_t\exp\left\{-\frac{\kappa\kappa_tL\epsilon^2}{(t+1)^4(\log M)^{2t+2}}\right\},
		\end{equation}
		and $X_n\overset{\circ}{=}c$ denote
		\begin{equation}
			\mathbb{P}\left(\left|X_n-c \right|\leq\epsilon\right)\leq (t+1)^2KK_t'\exp\left\{-\frac{\kappa\kappa_t'L\epsilon^2}{(t+1)^5(\log M)^{2t+2}}\right\},
		\end{equation}
		for $0\leq\epsilon\leq1$. Given power allocation satisfying $P_\ell=\Theta(\frac{1}{L})$ and $T$ not related to $M$, for $L$ and $M$ large enough, the following results hold for $0\leq t\leq T$.
		\begin{enumerate}[(a)]
			\item For all $t\geq0$,
			\begin{equation}
				\mathbb{P}\left(\frac{1}{N}||\boldsymbol{\Delta}_{pt}||^2\geq\epsilon \right)\leq (t+1)KK_t'\exp\left\{\frac{-\kappa\kappa_t'L\epsilon}{(t+1)^3(\log M)^{2t+1}}\right\},
				\label{Pa1}
			\end{equation}
			\begin{equation}
				\begin{aligned}
					&\mathbb{P}\left(\frac{1}{L}\sum_{\ell=1}^L \left[\max_{j\in sec(\ell)}|\boldsymbol{\Delta}_{pt}|_j \right]^2\geq\epsilon\right)\\
					&\qquad\qquad\leq (t+1)KK_t'\exp\left\{\frac{-\kappa\kappa_t'L\epsilon}{(t+1)^3(\log M)^{2t+1}}\right\},
				\end{aligned}
				\label{Pa2}
			\end{equation}
			\begin{equation}
				\mathbb{P}\left(\frac{1}{n}||\boldsymbol{\Delta}_{qt}||^2\geq\epsilon \right)\leq (t+1)KK_t\exp\left\{\frac{-\kappa\kappa_tL\epsilon}{(t+1)^3(\log M)^{2t+2}}\right\}.
				\label{Qa1}
			\end{equation}
			
			\item For all $0\leq j\leq t$,
			\begin{equation}
				\begin{aligned}
					&\mathbb{P}\left(\left|\frac{1}{N}\boldsymbol{p}_j^T\boldsymbol{p}_t-\left[\boldsymbol{\Sigma}_{pt}\right]_{j+1,t+1}\right|\geq\epsilon\right)\\
					&\qquad\qquad\leq (t+1)^2KK_t'\exp\left\{-\frac{\kappa\kappa_t'L\epsilon^2}{(t+1)^4(\log M)^{2t+2}}\right\}
				\end{aligned}
				\label{Pc1}
			\end{equation}
			\begin{equation}
				\frac{1}{n}\boldsymbol{q}_j^T\boldsymbol{q}_t\overset{.}{=}[\boldsymbol{\Sigma}_{qt}]_{j+1,t+1}/\alpha.
				\label{Qc1}
			\end{equation}
			
			\item For all $0\leq i\leq t$ and $1\leq j\leq t+1$
			\begin{equation}
				\frac{1}{n}\boldsymbol{v}_i^T\boldsymbol{v}_t\overset{\circ}{=}[\boldsymbol{\Sigma}_{qt}]_{i+1,t+1}/\alpha,
				\label{Pd1}
			\end{equation}
			\begin{equation}
				\frac{1}{N}\boldsymbol{u}_j^T\boldsymbol{u}_{t+1}\overset{.}{=}[\boldsymbol{\Sigma}_{pt}]_{j+1,t+2},\quad \frac{1}{N}\boldsymbol{u}_j^T\boldsymbol{u}_0\overset{.}{=}0.
				\label{Qd1}
			\end{equation}
			
			\item For all $0\leq j\leq t$
			\begin{equation}
				\frac{1}{n}\boldsymbol{p}_j^T\boldsymbol{v}_t\overset{\circ}{=}0,\quad \frac{1}{n}\boldsymbol{p}_t^T\boldsymbol{v}_j\overset{\circ}{=}0,
				\label{Pe1}
			\end{equation}
			\begin{equation}
				\frac{1}{n}\boldsymbol{q}_j^T\boldsymbol{u}_{t+1}\overset{.}{=}0,\quad \frac{1}{n}\boldsymbol{q}_t^T\boldsymbol{u}_{j+1}\overset{.}{=}0.
				\label{Qe1}
			\end{equation}
			
			\item\begin{enumerate}[(i)]
				\item\begin{equation}
					\begin{aligned}
						&\mathbb{P}\left(\left|\frac{1}{N}||[\boldsymbol{B}^\perp_{\boldsymbol{C_{vt}}}]^T\boldsymbol{p}_t||^2-\rho_{pt}\right|\geq\epsilon\right)\\
						&\qquad\qquad\leq (t+1)^4KK_t'\exp\left\{-\frac{\kappa\kappa_t'L\epsilon^2}{(t+1)^7(\log M)^{2t+2}}\right\}.
					\end{aligned}
					\label{Pfi1}
				\end{equation}
				\begin{equation}
					\begin{aligned}
						&\mathbb{P}\left(\left|\frac{1}{n}||[\boldsymbol{B}^\perp_{\boldsymbol{C_{ut}}}]^T\boldsymbol{q}_t||^2-\rho_{qt}/\alpha\right|\geq\epsilon\right)\\
						&\qquad\qquad\leq (t+1)^4KK_t\exp\left\{-\frac{\kappa\kappa_tL\epsilon^2}{(t+1)^6(\log M)^{2t+2}}\right\},
					\end{aligned}
					\label{Qfi1}
				\end{equation}
				
				\item For all $1\leq i,j\leq2t+1$,
				\begin{equation}
					\begin{aligned}
						&\mathbb{P}\left(\boldsymbol{C}_{pt}^T\boldsymbol{C}_{pt}\text{ is singular} \right)\\
						&\qquad\qquad\leq (t+1)^4KK_t'\exp\left\{\frac{-\kappa\kappa_t'L\epsilon^2}{(t+1)^7(\log M)^{2t+2}}\right\},
					\end{aligned}
					\label{Pfii1}
				\end{equation}
				\begin{equation}
					\begin{aligned}
						&\mathbb{P}\left(\left|[(\boldsymbol{C}_{pt}^T\boldsymbol{C}_{pt})^{-1}]_{ij}-[\boldsymbol{\Pi}_{pt}]_{ij}\right|\geq\epsilon\Big|\boldsymbol{C}_{pt}^T\boldsymbol{C}_{pt}\text{ invertible}\right)\\
						&\qquad\qquad\leq (t+1)^4KK_t'\exp\left\{-\frac{\kappa\kappa_t'L\epsilon^2}{(t+1)^7(\log M)^{2t+2}}\right\}.
					\end{aligned}
					\label{Pfii2}
				\end{equation}
				For all $1\leq i,j\leq2(t+1)$,
				\begin{equation}
					\begin{aligned}
						&\mathbb{P}\left(\boldsymbol{C}_{q,t+1}^T\boldsymbol{C}_{q,t+1} \text{ is singular}\right)\\
						&\qquad \qquad\leq (t+1)^6KK_t\exp\left\{\frac{-\kappa\kappa_tL\epsilon^2}{(t+1)^6(\log M)^{2t+2}}\right\},
					\end{aligned}
					\label{Qfii1}
				\end{equation}
				\begin{equation}
					\begin{aligned}
						&\mathbb{P}\left(\left|\left[(\boldsymbol{C}_{q,t+1}^T\boldsymbol{C}_{q,t+1})^{-1}\right]_{ij}-
					\left[\boldsymbol{\Pi}_{q,t+1}\right]_{ij}\right|\geq\epsilon\Big|\boldsymbol{C}_{q,t+1}^T\boldsymbol{C}_{q,t+1}\right.\\
					&\left.\text{ invertible}\right)\leq (t+1)^6KK_t\exp\left\{\frac{-\kappa\kappa_tL\epsilon^2}{(t+1)^6(\log M)^{2t+2}}\right\}.
					\end{aligned}
					\label{Qfii2}
				\end{equation}
				\item\begin{equation}
					\begin{aligned}
						&\mathbb{P}\left(\left|\frac{1}{n}||[\boldsymbol{B}^\perp_{\boldsymbol{C_{pt}}}]^T\boldsymbol{v}_t||^2-\rho_{qt}/\alpha\right|\geq\epsilon\right)\\
						&\qquad \qquad\leq (t+1)^6KK_t'\exp\left\{\frac{-\kappa\kappa_t'L\epsilon^2}{(t+1)^9(\log M)^{2t+2}}\right\},
					\end{aligned}
					\label{Pfiii1}
				\end{equation}
				\begin{equation}
					\begin{aligned}
						&\mathbb{P}\left(\left|	\frac{1}{N}||[\boldsymbol{B}^\perp_{\boldsymbol{C_{q,t+1}}}]^T\boldsymbol{u}_{t+1}||^2-\rho_{p,t+1}\right|\geq\epsilon\right)\\
						&\qquad \qquad\leq (t+1)^6KK_t\exp\left\{\frac{-\kappa\kappa_tL\epsilon^2}{(t+1)^8(\log M)^{2t+2}}\right\}.
					\end{aligned}
					\label{Qfiii1}
				\end{equation}
				\item For all $1\leq i,j\leq2(t + 1)$,
				\begin{equation}
					\begin{aligned}
						&\mathbb{P}\left(\boldsymbol{C}_{v,t+1}^T\boldsymbol{C}_{v,t+1}\text{ is singular}\right)\\
						&\qquad \qquad \leq (t+1)^6KK_t'\exp\left\{\frac{-\kappa\kappa_t'L\epsilon^2}{t^9(\log M)^{2t+2}}\right\},
					\end{aligned}
					\label{Pfiv1}
				\end{equation}
				\begin{equation}
					\begin{aligned}
					&\mathbb{P}\left(\left|\left[(\frac{1}{N}\boldsymbol{C}_{v,t+1}^T\boldsymbol{C}_{v,t+1})^{-1}\right]_{ij}- \left[\boldsymbol{\Pi}_{qt}\right]_{ij}\right|\geq\epsilon\Big|\boldsymbol{C}_{v,t+1}^T\boldsymbol{C}_{v,t+1}\right.\\
					&\left.\text{ invertible}\right)\leq (t+1)^6KK_t'\exp\left\{\frac{-\kappa\kappa_t'L\epsilon^2}{(t+1)^9(\log M)^{2t+2}}\right\}.
				\end{aligned}
					\label{Pfiv2}
				\end{equation}
				For all $1\leq i,j\leq2t+3$,
				\begin{equation}
					\begin{aligned}
						&\mathbb{P}\left(\boldsymbol{C}_{u,t+1}^T\boldsymbol{C}_{u,t+1}\text{ is singular}\right)\\
						&\qquad\leq (t+1)^6KK_t\exp\left\{\frac{-\kappa\kappa_tL\epsilon^2}{(t+1)^8(\log M)^{2t+2}}\right\},
					\end{aligned}
					\label{Qfiv1}
				\end{equation}
				\begin{equation}
					\begin{aligned}
						&\mathbb{P}\left(\left|\left[(\frac{1}{N}\boldsymbol{C}_{u,t+1}^T\boldsymbol{C}_{u,t+1})^{-1}\right]_{ij}-\left[\boldsymbol{\Pi}_{p,t+1}\right]_{ij}\right|\geq\epsilon\right.\\
						&\qquad\qquad\Big| \left.\boldsymbol{C}_{u,t+1}^T\boldsymbol{C}_{u,t+1}\text{ invertible}\right)\\
						&\qquad\leq (t+1)^6KK_t\exp\left\{\frac{-\kappa\kappa_tL\epsilon^2}{(t+1)^8(\log M)^{2t+2}}\right\}.
					\end{aligned}
					\label{Qfiv2}
				\end{equation}
			\end{enumerate}
			\item For all $t\geq0$,
			\begin{equation}
				\begin{aligned}
					&\mathbb{P}\left(\frac{1}{L}\sum_{\ell=1}^L\left[\max_{j\in sec(\ell)}(
					\boldsymbol{p}_t)_j^2 \right]\geq6\bar{\gamma}_{10}^{-1}\log M+\epsilon \right)\\
					&\qquad \qquad \leq (t+1)^2KK_t'\exp\left\{\frac{-\kappa\kappa_t'L\epsilon^2}{(t+1)^5(\log M)^{2t+1}}\right\},\label{Pg1}
				\end{aligned}
			\end{equation}
			\begin{equation}
				\begin{aligned}
					&\mathbb{P}\left(\frac{1}{L}\sum_{\ell=1}^L\left[\max_{j\in sec(\ell)}(\boldsymbol{v}_t)_j^2 \right]\geq c\log M+\epsilon \right) \\
					&\qquad \qquad \leq (t+1)^2KK_t'\exp\left\{\frac{-\kappa\kappa_t'L\epsilon^2}{(t+1)^5(\log M)^{2t+1}}\right\},\label{Pg2}
				\end{aligned}
			\end{equation}
			where $c$ is a constant not related to $L$ and $M$.
		\end{enumerate}
	\end{lemma}
	\begin{proof}
		We denote \eqref{Pa1}, \eqref{Pa2}, \eqref{Pc1}, \eqref{Pd1}, \eqref{Pe1}, \eqref{Pfi1}, \eqref{Pfii1}, \eqref{Pfii2}, \eqref{Pfiii1}, \eqref{Pfiv1}, \eqref{Pfiv2}, \eqref{Pg1}, \eqref{Pg2} as $P_t(a)-P_t(f)$ and \eqref{Qa1}, \eqref{Qc1}, \eqref{Qd1}, \eqref{Qe1}, \eqref{Qfi1}, \eqref{Qfii1}, \eqref{Qfii2}, \eqref{Qfiii1}, \eqref{Qfiv1}, \eqref{Qfiv2} as $Q_t(a)-Q_t(f)$. We will prove the lemma by induction.
		
		$P_0(a)$: \eqref{Pa1} follows from \cite{rush2022finite}. From Lemma \ref{lemma:distribution},
		\begin{equation*}
			\begin{aligned}
				&\mathbb{P}\left(\frac{1}{L}\sum_{\ell=1}^L\left[\max_{j\in sec(\ell)}|\boldsymbol{\Delta}_{p0}|_j \right]^2\geq\epsilon \right)\\
				&=\mathbb{P}\left(\frac{1}{L}\sum_{\ell=1}^L \left[\max_{j\in sec(\ell)}\left|\left(\frac{||\boldsymbol{u}_0||}{||\boldsymbol{Z}_{p0}||}-\sqrt{\rho_{p0}}\right)\boldsymbol{B}_{\boldsymbol{C}_{v0}}\boldsymbol{Z}_{p0} \right| \right]^2\geq\epsilon \right)\\
				&\leq\mathbb{P}\left(\left|\frac{||\boldsymbol{u}_0||}{||\boldsymbol{Z}_{p0}||}-\sqrt{\rho_{p0}}\right|\geq\frac{\epsilon}{6\log M}\right)\\
				&\qquad +\mathbb{P}\left(\frac{1}{L}\sum_{\ell=1}^L \left[\max_{j\in sec(\ell)}\left([\boldsymbol{B}^\perp_{\boldsymbol{C}_{v0}}\boldsymbol{Z}_{p0}]_j\right)^2\right]\geq6\log M\right).
			\end{aligned}
		\end{equation*}
		As $\rho_{p0}=\mathbb{E}[P_0^2]=\mathbb{E}[[\boldsymbol{u}_0]_i^2]$, the first term concentrates. For the second term, $\boldsymbol{B}^\perp_{\boldsymbol{C}_{v0}}\boldsymbol{Z}_{p0}=[\boldsymbol{B}^\perp_{\boldsymbol{C}_{v0}},\boldsymbol{\boldsymbol{B}_{\boldsymbol{C}_{v0}}}][\boldsymbol{Z}_{p0};\boldsymbol{\breve{Z}}_{p0}]-\boldsymbol{\boldsymbol{B}_{\boldsymbol{C}_{v0}}}\boldsymbol{\breve{Z}}_{p0}$, where $\boldsymbol{\boldsymbol{B}_{\boldsymbol{C}_{v0}}}\boldsymbol{\breve{Z}}_{p0}=\frac{\boldsymbol{v}_0}{||\boldsymbol{v}_0||}\breve{Z}_{p0}$; thus,
		\begin{equation*}
			\begin{aligned}
				&\mathbb{P}\left(\frac{1}{L}\sum_{\ell=1}^L\left[\max_{j\in sec(\ell)}\left([\boldsymbol{B}^\perp_{\boldsymbol{C}_{v0}}\boldsymbol{Z}_{p0}]_j\right)^2\right]\geq6\log M\right)\leq\\ &\mathbb{P}\left(\frac{1}{L}\sum_{\ell=1}^L\left[\max_{j\in sec(\ell)}(\bar{\boldsymbol{Z}}_{p0})^2\right]\geq3\log M\right)+\mathbb{P}\left(\frac{\breve{Z}_{p0}^2}{L}\geq3\log M\right),
			\end{aligned}
		\end{equation*}
		concentrates by Lemma \ref{concentration:max}.
		
		$P_0(b)$: 
		\begin{equation*}
			\begin{aligned}
				&\mathbb{P}\left( \left|\frac{\boldsymbol{p}_0^T\boldsymbol{p}_0}{N}-[\boldsymbol{\Sigma}_{p0}]_{1,1} \right|\geq\epsilon\right) \\
				&\leq\mathbb{P}\left( \left|\frac{\smash[t]{\overset{*}{\boldsymbol{p}_0}}^T\smash[t]{\overset{*}{\boldsymbol{p}_0}}}{N}-[\boldsymbol{\Sigma}_{pt}]_{1,1} \right|\geq\frac{\epsilon}{2}\right)+\mathbb{P}\left(\frac{||\boldsymbol{\Delta}_{pt}||^2}{N}\geq\frac{\epsilon^2}{4\bar{\gamma}_{1t}^{-1}}\right)\\
				&\leq K\exp\left\{-\kappa L\epsilon\right\},
			\end{aligned}
		\end{equation*}
		where we bound the first term because $\smash[t]{\overset{*}{\boldsymbol{p}_0}}$ is Gaussian with $\mathbb{E}[[\smash[t]{\overset{*}{\boldsymbol{p}_0}}]_i^2]=\bar{\gamma}_{10}^{-1}$, and bound the second term by $P_0(a)$.
		
		$P_0(c)$: By arguments in \cite{rush2018error}, under $P_0(a)$, we have $$\mathbb{P}(|\frac{1}{n}f_p(\boldsymbol{p}_0,\boldsymbol{\omega}_p,\bar{\gamma}_{10})^T\boldsymbol{p}_0-\bar{\alpha}_{10}[\boldsymbol{\Sigma}_{p0}]_{1,1}|\geq\epsilon)\leq K\exp\left\{-\frac{\kappa L\epsilon^2}{(\log M)^2}\right\}$$ and
		\begin{equation*}
			\begin{aligned}
				&\mathbb{P}(|\frac{1}{n}||f_p(\boldsymbol{p}_0,\boldsymbol{\omega}_p,\bar{\gamma}_{10})||^2-
				\mathbb{E}[f_p(P_0,W_p,\bar{\gamma}_{10})^2]\geq\epsilon)\\
				&\qquad\qquad\leq K\exp\left\{-\frac{\kappa L\epsilon^2}{(\log M)^2}\right\}.
			\end{aligned}
		\end{equation*} 
		Thus
		\begin{equation*}
			\begin{aligned}
				&\mathbb{P}\left(\left|\frac{1}{n}||\boldsymbol{v}_0||^2-[\boldsymbol{\Sigma}_{q0}]_{1,1} \right|\geq\epsilon \right)\\
				&\leq\mathbb{P}\Big(\Big|\frac{\bar{\alpha}_{10}}{n}f_p(\boldsymbol{p}_0,\boldsymbol{\omega}_p,\bar{\gamma}_{10})^T\boldsymbol{p}_0 -(\bar{\alpha}_{10})^2[\boldsymbol{\Sigma}_{pt}]_{1,1}\Big|\geq\frac{\epsilon(1-\bar{\alpha}_{10})^2}{4} \Big)\\
				&+\mathbb{P}\Big(\Big|\frac{||f_p(\boldsymbol{p}_0,\boldsymbol{\omega}_p,\bar{\gamma}_{1t})||^2}{n} -\mathbb{E}[f_p(P_0,W_p,\bar{\gamma}_{10})^2] \Big|\geq \frac{\epsilon(1-\bar{\alpha}_{10})^2}{4}\Big)
				\\&
				+\mathbb{P}\left(\left|\frac{1}{N}||\boldsymbol{p}_0||^2-[\boldsymbol{\Sigma}_{p0}]_{1,1}\right|\geq\frac{\alpha\epsilon(1-\bar{\alpha}_{1t})^2}{4(\bar{\alpha}_{10})^2}\right),\\
				&\leq K\exp\left\{-\frac{\kappa L\epsilon^2}{(\log M)^2}\right\}
			\end{aligned}
		\end{equation*}
		also concentrates as $\bar{\alpha}_{10}/\alpha$ is upper bounded, where the last term is from $P_0(b)$.
		
		$P_0(d)$: Concentration is shown similarly to $P_t(c)$ using
		\begin{equation*}
			\begin{aligned}
				&\mathbb{P}\left(\frac{|\boldsymbol{p}_0^T\boldsymbol{v}_0|}{n} \geq\epsilon \right)\leq
				\mathbb{P}\left(\left|\frac{\bar{\alpha}_{10} \boldsymbol{p}_0^T\boldsymbol{p}_0}{n}-[\boldsymbol{\Sigma}_{p0}]_{1,1}\right|\geq\frac{\epsilon[1-\bar{\alpha}_{10}]}{2}\right) \\
				& + \mathbb{P}\left(\left|\frac{f_p(\boldsymbol{p}_0,\boldsymbol{\omega}_p,\bar{\gamma}_{1t})^T\boldsymbol{p}_0}{n}-\bar{\alpha}_{10}[\boldsymbol{\Sigma}_{p0}]_{1,1}\right| \geq\frac{\epsilon[1-\bar{\alpha}_{10}]}{2}\right).
			\end{aligned}
		\end{equation*}
		
		$P_0(e)$: This follows from \cite{rush2022finite}.
		
		$P_0(f)$: To prove \eqref{Pg1}, from Lemma \ref{lemma:sum} we have
		\begin{equation}
			\begin{aligned}
				&\mathbb{P}\left(\frac{1}{L}\sum_{\ell=1}^L \left[\max_{j\in sec(\ell)}(
				\boldsymbol{p}_0)_j^2 \right]\geq6\bar{\gamma}_{10}^{-1}\log M+\epsilon\right)\\
				&\leq \mathbb{P}\left(\frac{1}{L}\sum_{\ell=1}^L \left[\max_{j\in sec(\ell)}([
				\smash[t]{\overset{*}{\boldsymbol{p}}}_0]_j)^2\right]\geq3\bar{\gamma}_{10}^{-1}\log M \right)\\
				&+\mathbb{P}\left(\frac{1}{L}\sum_{\ell=1}^L \left[\max_{j\in sec(\ell)} ([\boldsymbol{\Delta}_{p0}]_j)^2 \right]\geq \frac{\epsilon}{2}\right)\\
				&\leq K\exp\left\{-\frac{\kappa L\epsilon^2}{\log M}\right\}.
			\end{aligned}
		\end{equation}
		The first term concentrates by Lemma \ref{concentration:max} and the second term concentrates by $P_0(a)$.
		
		For \eqref{Pg2}, we have
		\begin{equation}
			\begin{aligned}
				&\mathbb{P}\left(\frac{1}{L}\sum_{\ell=1}^L \left[\max_{j\in sec(\ell)}
				([\boldsymbol{v}_t]_j)^2 \right]\geq c\log M+\epsilon \right)\leq\\
				&\mathbb{P}\left(\frac{1}{L}\sum_{\ell=1}^L \left[\max_{j\in sec(\ell)}
				([f_p(\boldsymbol{p}_j,\boldsymbol{\omega}_p,\bar{\gamma}_{1t})]_j)^2  \right]\geq\frac{c}{4}\log M[1-\bar{\alpha}_{1t}]^2  \right)\\
				&+\mathbb{P}\left(\frac{1}{L}\sum_{\ell=1}^L\left[\max_{j\in sec(\ell)}([\bar{\alpha}_{1t}\boldsymbol{p}_t]_j)^2 \right]\geq\left[\frac{c}{4}\log M+\epsilon\right][1-\bar{\alpha}_{1t}]^2  \right)
			\end{aligned}
		\end{equation}
		from Lemma \ref{lemma:sum}. From its definition, we have $\frac{1}{L}||f_p(\boldsymbol{p}_j,\boldsymbol{\omega}_p,\bar{\gamma}_{1t})||^2\leq c'\log M$, and thus the first term is $0$. The second term is upper bounded by $ K\exp\left\{-\frac{\kappa L\epsilon^2}{\log M}\right\}$ from \eqref{Pg1}.
		
		$Q_0(a)$: From Lemma \ref{lemma:distribution} we have
		\begin{equation*}
			\begin{aligned}
				\boldsymbol{\Delta}_{q0}&=\frac{\boldsymbol{p}_0^T\boldsymbol{v}_0}{||\boldsymbol{p}_0||^2}\boldsymbol{u}_0 -\sqrt{\rho_{q0}}\boldsymbol{B}_{\boldsymbol{C}_{u0}}\breve{\boldsymbol{Z}}_{q0} \\
				&\qquad  +\left[\frac{||[\boldsymbol{B}^\perp_{C_{p0}}]^T\boldsymbol{v}_0||}{||\boldsymbol{Z}_{q0}||}-\sqrt{\rho_{q0}}\right]\boldsymbol{B}^\perp_{\boldsymbol{C}_{u0}}\boldsymbol{Z}_{q0},
			\end{aligned}
		\end{equation*}
		and therefore
		\begin{equation}
			\begin{aligned}
				&\mathbb{P}\left(\frac{1}{n}||\boldsymbol{\Delta}_{q0}||^2\geq\epsilon\right)\leq
				\mathbb{P}\left(\frac{1}{n}\frac{\boldsymbol{p}_0^T\boldsymbol{v}_0}{||\boldsymbol{p}_0||^2}\geq\frac{\epsilon}{9}\right)\\
				&+\mathbb{P}\left(\left|\frac{||[\boldsymbol{B}^\perp_{C_{p0}}]^T\boldsymbol{v}_0||}{||\boldsymbol{Z}_{q0}||}-\sqrt{\rho_{q0}}\right|\frac{||\boldsymbol{Z}_{q0}||}{\sqrt{n}}\geq\frac{\sqrt{\epsilon}}{3}\right)\\
				&+\mathbb{P}\left(\sqrt{\frac{\rho_{q0}}{n}}||\breve{Z}_{q0}|| \geq\frac{\sqrt{\epsilon}}{3}\right).
			\end{aligned}
		\end{equation}
		Label the right side as $T_1-T_3$. $T_1$ is upper bounded by $K\exp\left\{\frac{-\kappa L\epsilon}{(\log M)^2}\right\}$ according to $P_0(b)$ and $P_0(c)$. $T_3$ is upper bounded by $K\exp\left\{-\kappa L\epsilon\right\}$ according to Lemma \ref{concentration:chi}. For $T_2$,
		\begin{equation}
			\begin{aligned}
				T_2\leq&\mathbb{P}\left(\left|\frac{||[\boldsymbol{B}_{\boldsymbol{C}_{p0}}^\perp]^T\boldsymbol{u}_0||}{\sqrt{n}}-\sqrt{\frac{\rho_{q0}}{\alpha}}\right|\geq\frac{\sqrt{\epsilon}}{6}\right)\\
				&+\mathbb{P}\left(\left|\frac{||\boldsymbol{Z}_{q0}||}{\sqrt{n}}-1\right|\geq\frac{\sqrt{\epsilon}}{6\max\{1,\sqrt{\rho_{q0}/\alpha}\}}\right)\\
				&\leq K\exp\left\{\frac{-\kappa L\epsilon}{(\log M)^2}\right\},
			\end{aligned}
		\end{equation}
		where the first term is from $P_t(e)(iii)$ and the second term is from Lemma \ref{concentration:chi}.
		
		$Q_0(b)-(e)$: They are proved in the same way as $Q_t(b)-(e)$.
		
		$P_t(a)$: 
		First, denote $\boldsymbol{\mu}_{pt}:=(\boldsymbol{C}_{vt}^T\boldsymbol{C}_{vt})^{-1}\boldsymbol{C}_{qt}^T\boldsymbol{u}_t$. Then, 
		\begin{equation*}
			\begin{aligned}
				\boldsymbol{\Delta}_{pt}&=\boldsymbol{C}_{vt}\left(\boldsymbol{\mu}_{pt}-\left[\begin{matrix}
					\boldsymbol{\beta}_{pt}\\\boldsymbol{0}
				\end{matrix}\right]\right) -\sqrt{\rho_{pt}}\boldsymbol{B}_{\boldsymbol{C}_{vt}}\breve{\boldsymbol{Z}}_{pt} \\
				&\qquad  +
				\left[\frac{||[\boldsymbol{B}^\perp_{C_{qt}}]^T\boldsymbol{u}_t||}{||\boldsymbol{Z}_{pt}||}-\sqrt{\rho_{pt}}\right]\boldsymbol{B}^\perp_{\boldsymbol{C}_{vt}}\boldsymbol{Z}_{pt} .
			\end{aligned}
		\end{equation*}
		Therefore, \eqref{Pa1} follows from \cite{rush2022finite}, and
		\begin{equation}
			\begin{aligned}
				&\mathbb{P}\left(\frac{1}{L}\sum_{\ell=1}^L\left[\max_{j\in sec(\ell)}|\boldsymbol{\Delta}_{pt}|_j\right]^2\geq\epsilon\right)\\
				&\leq\sum_{i=1}^t\mathbb{P}\left(\left|[\boldsymbol{\mu}_{pt}]_i-[\boldsymbol{\beta}_{pt}]_i \right|\frac{1}{L}\sum_{\ell=1}^L\left[\max_{j\in sec(\ell)}([\boldsymbol{C}_{vt}]_{(j,i)})^2 \right]\geq\epsilon'\right)\\
				&+\sum_{i=t+1}^{2t}\mathbb{P}\left(\left|[\boldsymbol{\mu}_{pt}]_i\frac{1}{L}\sum_{\ell=1}^L \left[\max_{j\in sec(\ell)}([\boldsymbol{C}_{vt}]_{(j,i)})^2 \right]\right|\geq\epsilon'\right)\\
				&+\mathbb{P}\left(\left|\frac{||[\boldsymbol{B}^\perp_{C_{qt}}]^T\boldsymbol{u}_t||}{||\boldsymbol{Z}_{pt}||}-\sqrt{\rho_{pt}}\right|
				\frac{1}{L}\sum_{\ell=1}^L \left[\max_{j\in sec(\ell)}([\boldsymbol{B}^\perp_{\boldsymbol{C}_{vt}}\boldsymbol{Z}_{pt}]_j)^2\right]\geq\epsilon'\right)\\
				&+\mathbb{P}\left( \frac{\sqrt{\rho_{pt}}}{L}\sum_{\ell=1}^L\left[\max_{j\in sec(\ell)}([\boldsymbol{\boldsymbol{B}_{\boldsymbol{C}_{vt}}\breve{Z}}_{pt}]_j)^2 \right] \geq\epsilon' \right),
			\end{aligned}
		\end{equation}
		where $\epsilon'=\frac{\epsilon}{9t^2}$ by Lemma \ref{lemma:sum}. Label the right side as $T_1-T_4$. 
		\begin{equation*}
			\begin{aligned}
				T_1\leq&\sum_{i=1}^t\left[\mathbb{P}\left(\left[\boldsymbol{\mu}_{pt} \right]_i-[\boldsymbol{\beta}_{pt}]_i\geq\frac{\bar{\gamma}_{10}\epsilon'}{4\log M}\right)\right.\\
				&\left.+\mathbb{P}\left(\frac{1}{L}\sum_{\ell=1}^L \left[\max_{j\in \text{sec}(\ell)}([\boldsymbol{p}_{i-1}]_j)^2 \right] \geq4\bar{\gamma}_{10}^{-1}\log M \right)\right]\\
				&\leq tKK_t'\exp\left\{\frac{-\kappa\kappa_t'L\epsilon}{t^2(\log M)^{2t+1}}\right\}.
			\end{aligned}
		\end{equation*}
		The first term concentrates by \cite[Lemma 7]{rush2022finite} and the second term concentrates by $P_1(f)-P_{t-1}(f)$. Then consider $T_2$,
		\begin{equation*}
			\begin{aligned}
				T_2\leq&\sum_{i=t+1}^{2t}\left[\mathbb{P}\left([\boldsymbol{\mu}_{pt}]_i\geq\frac{\epsilon'}{2c\log M}\right)\right.\\
				&\left.+\mathbb{P}\left(\frac{1}{L}\sum_{\ell=1}^L \left[\max_{j\in sec(\ell)}([\boldsymbol{v}_{i-1}]_j)^2 \right]\geq2c\log M \right)\right]\\
				&\leq tKK_t'\exp\left\{\frac{-\kappa\kappa_t'L\epsilon}{t^2(\log M)^{2t+1}}\right\}.
			\end{aligned}
		\end{equation*}
		The first term concentrates by Lemma 7 in \cite{rush2022finite} and the second term concentrates by $P_1(f)-P_{t-1}(f)$. Next we consider $T_4$.
		\begin{equation*}
			\begin{aligned}
				T_4\leq&\mathbb{P}\left(\sqrt{\rho_{pt}}\frac{1}{L}||\boldsymbol{\breve{Z}}_{pt}||^2\geq\epsilon'  \right)\leq2t\mathbb{P}\left(\breve{Z}_{pt}\geq\frac{L\epsilon'}{t\sqrt{\rho_{pt}}} \right)
				\\&\leq tK\exp\left\{\frac{-\kappa L\epsilon}{t^3(\log M)^{2t+1}}\right\}.
			\end{aligned}
		\end{equation*}
		concentrates, where the first inequality is from 
		$$\sum_{\ell=1}^L \left[\max_{j\in sec(\ell)}([{\boldsymbol{B}_{\boldsymbol{C}_{vt}} \breve{Z}}_{pt}]_j)^2 \right]\leq||\boldsymbol{B}_{\boldsymbol{C}_{vt}}\boldsymbol{\breve{Z}}_{pt}||^2=||\boldsymbol{\breve{Z}}_{pt}||^2.$$ 
		Finally consider $T_3$. 
		\begin{equation*}
			\begin{aligned}
				T_3\leq&\mathbb{P}\left(\left|\frac{||[\boldsymbol{B}^\perp_{C_{qt}}]^T\boldsymbol{u}_t||}{||\boldsymbol{Z}_{pt}||}-\sqrt{\rho_{pt}}\right|\geq\frac{\epsilon'}{6\log M}\right)\\
				&+\mathbb{P}\left(\frac{1}{L}\sum_{\ell=1}^L\left[\max_{j\in sec(\ell)}\left([\boldsymbol{B}^\perp_{\boldsymbol{C}_{vt}}\boldsymbol{Z}_{pt}]_j \right)^2 \right]\geq6\log M\right).
			\end{aligned}
		\end{equation*}
		The first term is upper bounded by $KK_t'\exp\left\{\frac{-\kappa\kappa_t'L\epsilon}{t^2(\log M)^{2t+1}}\right\}$ according to $Q_{t-1}(f)(iii)$. For the second term, $\boldsymbol{B}^\perp_{\boldsymbol{C}_{vt}}\boldsymbol{Z}_{pt}=[\boldsymbol{B}^\perp_{\boldsymbol{C}_{vt}},\boldsymbol{\boldsymbol{B}_{\boldsymbol{C}_{vt}}}][\boldsymbol{Z}_{pt};\boldsymbol{\breve{Z}}_{pt}]-\boldsymbol{\boldsymbol{B}_{\boldsymbol{C}_{vt}}}\boldsymbol{\breve{Z}}_{pt}$; thus,
		\begin{equation}
			\begin{aligned}
				&\mathbb{P}\left(\frac{1}{L}\sum_{\ell=1}^L\left[\max_{j\in sec(\ell)}\left(\left[\boldsymbol{B}^\perp_{\boldsymbol{C}_{vt}}\boldsymbol{Z}_{pt} \right]_j\right)^2\right]\geq6\log M\right)\\
				&\leq \mathbb{P}\left(\frac{1}{L}\sum_{\ell=1}^L \left [\max_{j\in sec(\ell)}(\bar{\boldsymbol{Z}}_{pt})^2 \right]\geq3\log M\right)\\	
				&+\mathbb{P}\left(\frac{1}{L}\sum_{\ell=1}^L \left[\max_{j\in sec(\ell)}\left(\left[\boldsymbol{\boldsymbol{B}_{\boldsymbol{C}_{vt}}\breve{Z}}_{pt}\right]_j\right)^2\right]\geq3\log M\right)\\
				&\leq K\exp\{-\kappa L\log M\} +tK\exp\left\{\frac{-L\epsilon}{t^3(\log M)^{2t+1}}\right\}
			\end{aligned}
		\end{equation}
		concentrates by Lemma \ref{concentration:max} and T4.
		
		$P_t(b)$:
		\begin{equation*}
			\begin{aligned}
				&\mathbb{P}\left(\left|\frac{\boldsymbol{p}_j^T\boldsymbol{p}_t}{N}-[\boldsymbol{\Sigma}_{pt}]_{j+1,t+1}\right|\geq\epsilon \right)
				\\
				&\leq \mathbb{P}\left(\left|\frac{\smash[t]{\overset{*}{\boldsymbol{p}_j}}^T\smash[t]{\overset{*}{\boldsymbol{p}_t}}}{N}-[\boldsymbol{\Sigma}_{pt}]_{j+1,t+1}\right|\geq\frac{\epsilon}{2}\right)\\
				&+\mathbb{P}\left(\frac{||\smash[t]{\overset{*}{\boldsymbol{p}_j}}||^2}{N}+\frac{||\smash[t]{\overset{*}{\boldsymbol{p}_t}}||^2}{N} \geq \frac{2}{\bar{\gamma}_{1t}}\right)\\
				&+\mathbb{P}\Big(\sum_{r=0}^j[c_{pj}]_r^2 \sum_{r=0}^j\frac{||\boldsymbol{\Delta}_{pr}||^2}{N} \\
				& \qquad \qquad +  \sum_{r=0}^t[c_{pt}]_r^2\sum_{r=0}^t\frac{||\boldsymbol{\Delta}_{pt}||^2}{N} \geq\frac{\epsilon^2}{4\bar{\gamma}_{1t}^{-1}}\Big).
			\end{aligned}
		\end{equation*}
		The first and the second terms are upper bounded by $K\exp\left\{-L\epsilon\right\}$ because $\smash[t]{\overset{*}{\boldsymbol{p}_t}}$ is Gaussian with $\mathbb{E}[[\smash[t]{\overset{*}{\boldsymbol{p}_t}}]_i^2]\geq\bar{\gamma}_{1t}^{-1}$ for $t\geq0$. For the last term (denoted as $T_3$), we first observe $\bar{\gamma}_{1t}=\mathbb{E}[([\boldsymbol{p}_t]_i)^2]=\sum_{r=0}^t\rho_{pr}([c_{pt}]_r)^2$; therefore, 
		\begin{equation}
			\sum_{r=0}^j[c_{pt}]_r^2\leq\frac{\bar{\gamma}_{1t}}{\min_{0\leq i\leq t} \{\rho_{pt} \}}. \label{c_{pt} bound}
		\end{equation}
		Thus, we have
		\begin{equation*}
			\begin{aligned}
				T_3\leq&\sum_{r=0}^j\mathbb{P}\Big( \frac{||\boldsymbol{\Delta}_{pr}||^2}{N}\geq\frac{\epsilon^2}{4j\sum_{r=0}^j[c_{pt}]_r^2\bar{\gamma}_{1t}^{-1}}\Big)\\
				&+\sum_{r=0}^t\mathbb{P}\Big( \frac{||\boldsymbol{\Delta}_{pr}||^2}{N}\geq\frac{\epsilon^2}{4t\sum_{r=0}^t[c_{pt}]_r^2\bar{\gamma}_{1t}^{-1}}\Big)\\
				&\leq (t+1)^2KK_t'\exp\left\{\frac{-\kappa\kappa_t'L\epsilon}{(t+1)^4(\log M)^{2t+1}}\right\},
			\end{aligned}
		\end{equation*} 
		where the last inequality is from $P_j(a)$, $P_t(a)$ and \eqref{c_{pt} bound}.
		
		$P_t(c)$: By arguments in \cite{rush2018error}, under $P_t(a)$, we have concentration
		\begin{equation*}
			\begin{aligned}
				&\mathbb{P}\left(\left|\frac{1}{n}f_p(\boldsymbol{p}_i,\boldsymbol{\omega}_p,\bar{\gamma}_{1t})^T\boldsymbol{p}_t-\bar{\alpha}_{1t}[\boldsymbol{\Sigma}_{pt}]_{i+1,t+1}\right|\geq\epsilon\right)\\
				&\quad \leq (t+1)^2KK_t'\exp\left\{-\frac{\kappa\kappa_t'L\epsilon^2}{(t+1)^5(\log M)^{2t+2}}\right\}
			\end{aligned}
		\end{equation*}
	 	and 
		\begin{equation*}
			\begin{aligned}
				& \mathbb{P}\Big(\Big|\frac{1}{n}f_p(\boldsymbol{p}_i,\boldsymbol{\omega}_p,\bar{\gamma}_{1i})^T f_p(\boldsymbol{p}_t,\boldsymbol{\omega}_p,\bar{\gamma}_{1t}) \\
				&\qquad -
				\mathbb{E}\left[f_p(P_i,W_p,\bar{\gamma}_{1t})f_p(P_t,W_p,\bar{\gamma}_{1t})\right] \Big|\geq\epsilon \Big)\\
				&\quad\leq (t+1)^2KK_t'\exp\left\{-\frac{\kappa\kappa_t'L\epsilon^2}{(t+1)^5(\log M)^{2t+2}}\right\}
			\end{aligned}
		\end{equation*}
		for $i\leq t$. Thus
		\begin{equation*}
			\begin{aligned}
				&\mathbb{P}\left(\left|\frac{1}{n}\boldsymbol{v}_i^T\boldsymbol{v}_t-[\boldsymbol{\Sigma}_{qt}]_{i+1,t+1}\right|\geq\epsilon\right)\\
				&\leq \mathbb{P}\Big(\Big|\frac{\bar{\alpha}_{1t}}{n}f_p(\boldsymbol{p}_i,\boldsymbol{\omega}_p,\bar{\gamma}_{1t})^T\boldsymbol{p}_t -(\bar{\alpha}_{1t})^2[\boldsymbol{\Sigma}_{pt}]_{i+1,t+1}\Big|\geq\frac{\epsilon(1-\bar{\alpha}_{1t})^2}{4}\Big)\\
				&+\mathbb{P}\Big(\Big|\frac{1}{n}f_p(\boldsymbol{p}_j,\boldsymbol{\omega}_p,\bar{\gamma}_{1t})^Tf_p(\boldsymbol{p}_t,\boldsymbol{\omega}_p,\bar{\gamma}_{1t})\\
				&\qquad -\mathbb{E}\left[f_p(P_i,W_p,\bar{\gamma}_{1t})f_p(P_t,W_p,\bar{\gamma}_{1t})\right] \Big|\geq\frac{\epsilon(1-\bar{\alpha}_{1t})^2}{4}\Big)\\
				&+\mathbb{P}\left(\left|\frac{1}{N}\boldsymbol{p}_i^T\boldsymbol{p}_t-[\boldsymbol{\Sigma}_{pt}]_{i+1,t+1} \right|\geq\frac{\alpha\epsilon(1-\bar{\alpha}_{1t})^2}{4(\bar{\alpha}_{1t})^2}\right)\\
				&\quad\leq (t+1)^2KK_t'\exp\left\{-\frac{\kappa\kappa_t'L\epsilon^2}{(t+1)^5(\log M)^{2t+2}}\right\}
			\end{aligned}
		\end{equation*}
		as $\bar{\alpha}_{1t}/\alpha$ is upper bounded, where the last term is bounded according to $P_t(b)$.
		
		$P_t(d)$: Concentration is shown similarly to $P_t(c)$ using
		\begin{equation*}
			\begin{aligned}
				&\mathbb{P}\left(\left|\frac{1}{n}\boldsymbol{p}_j^T\boldsymbol{v}_t \right|\geq\epsilon \right)
				\leq \mathbb{P}\left(\left|\frac{\bar{\alpha}_{1t}}{n}\boldsymbol{p}_j\boldsymbol{p}_t-[\boldsymbol{\Sigma}_{pt}]_{jt}\right|\geq\frac{\epsilon(1-\bar{\alpha}_{1t})}{2}\right)\\
				&+\mathbb{P} \left(\left|\frac{1}{n}f_p(\boldsymbol{p}_j,\boldsymbol{\omega}_p,\bar{\gamma}_{1t})^T\boldsymbol{p}_t-\bar{\alpha}_{1t}[\boldsymbol{\Sigma}_{pt}]_{jt}\right|\geq\frac{\epsilon(1-\bar{\alpha}_{1t})}{2}\right).
			\end{aligned}
		\end{equation*}

		$P_t(e)$: This follows from \cite{rush2022finite}.
		
		$P_t(f)$: To prove \eqref{Pg1}, we have
		\begin{equation*}
			\begin{aligned}
				&\mathbb{P}\left(\frac{1}{L}\sum_{\ell=1}^L \left[\max_{j\in sec(\ell)}(
				\boldsymbol{p}_t)_j^2 \right]\geq6\bar{\gamma}_{10}^{-1}\log M+\epsilon\right)\\
				&\leq \mathbb{P}\left(\frac{1}{L}\sum_{\ell=1}^L \left[\max_{j\in sec(\ell)}(
				\smash[t]{\overset{*}{\boldsymbol{p}}}_t)_j^2 \right]\geq3\bar{\gamma}_{10}^{-1}\log M\right)\\
				&+\mathbb{P}\left(\frac{1}{N}\sum_{r=0}^t[c_{pt}]_r^2\sum_{r=0}^t\left\|\frac{1}{L}\max_{j\in sec(\ell)}[\boldsymbol{\Delta}_{pr}]_j\right\|^2\geq \frac{\epsilon}{2}\right)\\
				&\leq K\exp\left\{-\kappa L\epsilon^2\right\}+(t+1)^2KK_t'\exp\left\{\frac{-\kappa\kappa_t'L\epsilon^2}{(t+1)^5(\log M)^{2t+1}}\right\}.
			\end{aligned}
		\end{equation*}
		The first inequality is from Lemma \ref{lemma:sum}. For the second inequality, the first term concentrates by Lemma \ref{concentration:max} and $\bar{\gamma}_{1t}^{-1}<\bar{\gamma}_{10}^{-1}$, and the second term concentrates due to $P_t(a)$ and \eqref{c_{pt} bound}.
		
		For \eqref{Pg2}, we have
		\begin{equation*}
			\begin{aligned}
				&\mathbb{P}\left(\frac{1}{L}\sum_{\ell=1}^L \left[\max_{j\in sec(\ell)}(
				\boldsymbol{v}_t)_j^2 \right]\geq c\log M+\epsilon \right)\\
				&\leq \mathbb{P}\left(\frac{1}{L}\sum_{\ell=1}^L \left[\max_{j\in sec(\ell)}
				f_p(\boldsymbol{p}_j,\boldsymbol{\omega}_p,\bar{\gamma}_{1t})_j^2\right]\geq\frac{c}{4}\log M(1-\bar{\alpha}_{1t})^2 \right)\\
				&+\mathbb{P}\left(\frac{1}{L}\sum_{\ell=1}^L \left[\max_{j\in sec(\ell)}(
				[\bar{\alpha}_{1t}\boldsymbol{p}_t)]_j)^2\right]\geq \left(\frac{c}{4}\log M+\epsilon\right) \left(1-\bar{\alpha}_{1t}\right)^2\right),
			\end{aligned}
		\end{equation*}
		from Lemma \ref{lemma:sum}. Next, from its definition, we have that $\frac{1}{L}||f_p(\boldsymbol{p}_j,\boldsymbol{\omega}_p,\bar{\gamma}_{1t})||^2\leq c'\log M$; thus, the first term is $0$. The second term is upper bounded by $(t+1)^2KK_t'\exp\left\{\frac{-\kappa\kappa_t'L\epsilon^2}{(t+1)^5(\log M)^{2t+1}}\right\}$ from \eqref{Pg1}.
		
		$Q_t(a)$:  Analogously to $P_t(a)$, denote $\boldsymbol{\mu}_{qt}:=([\boldsymbol{C}_{pt}]^T\boldsymbol{C}_{pt})^{-1}[\boldsymbol{C}_{qt}]^T\boldsymbol{v}_t$. Then, 
		\begin{equation*}
			\begin{aligned}
				\boldsymbol{\Delta}_{qt}&=\boldsymbol{C}_{ut}\left(\boldsymbol{\mu}_{qt}-\left[\begin{matrix}
					\boldsymbol{0}\\\boldsymbol{\beta}_{qt}
				\end{matrix}\right]\right) -\sqrt{\rho_{qt}}\boldsymbol{B}_{\boldsymbol{C}_{ut}}\breve{\boldsymbol{Z}}_{qt} \\
				&\qquad  +
				\left[\frac{||[\boldsymbol{B}^\perp_{C_{pt}}]^T\boldsymbol{v}_t||}{||\boldsymbol{Z}_{qt}||}-\sqrt{\rho_{qt}}\right]\boldsymbol{B}^\perp_{\boldsymbol{C}_{ut}}\boldsymbol{Z}_{qt},
			\end{aligned}
		\end{equation*}
		and therefore
		\begin{equation*}
			\begin{aligned}
				&\mathbb{P}\left(\frac{1}{n}||\boldsymbol{\Delta}_{qt}||^2\geq\epsilon\right)\leq\sum_{i=1}^{t+1}\mathbb{P}\left(\left|[\boldsymbol{\mu}_{qt}]_i\right|\frac{||[\boldsymbol{C}_{ut}]_{(\cdot,i)}||}{\sqrt{n}} \geq\frac{\sqrt{\epsilon}}{4(t+1)}\right)\\
				&+\sum_{i=t+2}^{2t+1}\mathbb{P}\left(\left|[\boldsymbol{\mu}_{qt}]_{i}-[\boldsymbol{\beta}_{qt}]_{i-t-1} \right|\frac{||[\boldsymbol{C}_{ut}]_{(\cdot,i)}||}{\sqrt{n}}\geq\frac{\sqrt{\epsilon}}{4(t+1)}\right)\\
				&+\mathbb{P}\left(\left|\frac{||[\boldsymbol{B}^\perp_{C_{pt}}]^T\boldsymbol{v}_t||}{||\boldsymbol{Z}_{qt}||}-\sqrt{\rho_{qt}}\right|\frac{||\boldsymbol{Z}_{qt}||}{\sqrt{n}}\geq\frac{\sqrt{\epsilon}}{4(t+1)}\right)\\
				&+\mathbb{P}\left(\sqrt{\frac{\rho_{qt}}{n}}||\breve{\boldsymbol{Z}}_{qt}|| \geq\frac{\sqrt{\epsilon}}{4(t+1)}\right).
			\end{aligned}
		\end{equation*}
		Label the right side as $T_1-T_4$.
		\begin{equation*}
			\begin{aligned}
				T_1\leq&\sum_{i=0}^{t}\left(\mathbb{P}\left(\left|\frac{||\boldsymbol{u}_i||}{\sqrt{N}}-(\boldsymbol{\Sigma}_{ut})_{ii}\right|\geq\frac{\sqrt{\epsilon}}{8(t+1)}\right)\right.\\
				&\left.+\mathbb{P}\left(\left|\frac{1}{\sqrt{\alpha}}[\boldsymbol{\mu}_{qt}]_{i+1}\right|\geq\frac{\sqrt{\epsilon}}{8(t+1)\max\{1,(\Sigma_{ut})_{ii}\}}\right)\right).
			\end{aligned}
		\end{equation*}
		The first term is upper bounded by $ (t+1)KK_t\exp\left\{\frac{-\kappa\kappa_tL\epsilon}{(t+1)^2(\log M)^{2t+1}}\right\}$ according to $Q_0(c)-Q_{t-1}(c)$. For $1\leq i\leq t+1$,
		\begin{equation*}
			[\boldsymbol{\mu}_{qt}]_i=\sum_{j=1}^{t+1}\left[(\frac{1}{N}\boldsymbol{C}_{qt}^T\boldsymbol{C}_{qt})^{-1}\right]_{ij}\frac{1}{N}(\boldsymbol{p}_{j-1}^T\boldsymbol{v}_t),
		\end{equation*}
		and thus the second term is upper bounded by $(t+1)KK_t\exp\left\{\frac{-\kappa\kappa_tL\epsilon}{(t+1)^2(\log M)^{2t+1}}\right\}$ according to $Q_0(e)(ii)-Q_{t-1}(e)(ii)$ and $P_t(d)$. Then consider $T_2$,
		\begin{equation}
			\begin{aligned}
				T_2\leq&\sum_{i=t+2}^{2t}\left(\mathbb{P}\left(\left|\frac{||\boldsymbol{q}_{i'}||}{\sqrt{n}}-(\boldsymbol{\Sigma}_{vt})_{i'i'}\right|\geq\frac{\sqrt{\epsilon}}{8(t+1)}\right)\right.\\
				&\left.+\mathbb{P}\left(\left|[\boldsymbol{\mu}_{qt}]_i-[\boldsymbol{\beta}_{qt}]_{i'+1}\right|\geq\frac{\sqrt{\epsilon}}{8(t+1)\max\{1,(\Sigma_{vt})_{ii}\}}\right)\right),
			\end{aligned}
		\end{equation}
		where $i':=i-t-2$. The first term is upper bounded by $ (t+1)KK_t\exp\left\{\frac{-\kappa\kappa_tL\epsilon}{(t+1)^2(\log M)^{2t+1}}\right\}$ according to $Q_0(b)-Q_{t-1}(b)$. For $i+2\leq i\leq 2t+1$,
		\begin{equation*}
			[\boldsymbol{\mu}_{qt}]_i=\sum_{j=1}^{t}\left[(\frac{1}{N}\boldsymbol{C}_{qt}^T\boldsymbol{C}_{qt})^{-1}\right]_{ij}\frac{1}{N}(\boldsymbol{v}_{j-1}^T\boldsymbol{v}_t),
		\end{equation*}
		and thus the second term is upper bounded by $(t+1)KK_t\exp\left\{\frac{-\kappa\kappa_tL\epsilon}{(t+1)^2(\log M)^{2t+1}}\right\}$ according to $Q_0(e)(ii)-Q_{t-1}(e)(ii)$ and $P_t(c)$. Then consider $T_3$,
		\begin{equation*}
			\begin{aligned}
				T_3\leq&\mathbb{P}\left(\left|\frac{||[\boldsymbol{B}_{\boldsymbol{C}_{pt}}^\perp]^T\boldsymbol{v}_t||}{\sqrt{n}}-\sqrt{\frac{\rho_{qt}}{\alpha}}\right|\geq\frac{\sqrt{\epsilon}}{8(t+1)}\right)\\
				&+\mathbb{P}\left(\left|\frac{||\boldsymbol{Z}_{qt}||}{\sqrt{n}}-1\right|\geq\frac{\sqrt{\epsilon}}{8(t+1)\max\{1,\sqrt{\rho_{qt}/\alpha}\}}\right),\\
				&\leq (t+1)KK_t\exp\left\{\frac{-\kappa\kappa_tL\epsilon}{(t+1)^2(\log M)^{2t+1}}\right\}+K\exp\left\{-\kappa L\epsilon\right\},
			\end{aligned}
		\end{equation*}
		where the first term concentrates by $P_t(e)(iii)$ and the second term concentrates by Lemma \ref{concentration:chi}. Lastly consider $T_4$.
		\begin{equation}
			\begin{aligned}
				T_4\leq&\sum_{i=1}^{2t+1}\mathbb{P}\left([\breve{\boldsymbol{Z}}_{qt}]_i^2\geq\frac{N\epsilon}{32(t+1)^3(\rho_{qt}/\alpha)}\right)\\
				&\leq (t+1)KK_t\exp\left\{\frac{-\kappa\kappa_tL\epsilon}{(t+1)^3}\right\}
			\end{aligned}
		\end{equation}
		concentrates due to Lemma \ref{concentration:chi}.
		
		$Q_t(b)$: \begin{equation*}
			\begin{aligned}
				&\mathbb{P}\left(\left|\frac{\boldsymbol{q}_j^T\boldsymbol{q}_t}{n}-[\boldsymbol{\Sigma}_{qt}]_{j+1,t+1}\right|\geq\epsilon \right)
				\\
				&\leq \mathbb{P}\left(\left|\frac{\smash[t]{\overset{*}{\boldsymbol{q}_j}}^T\smash[t]{\overset{*}{\boldsymbol{q}_t}}}{n}-[\boldsymbol{\Sigma}_{qt}]_{j+1,t+1}\right|\geq\frac{\epsilon}{2}\right)+\mathbb{P}\left(\frac{||\smash[t]{\overset{*}{\boldsymbol{q}_j}}||^2}{n}+\frac{||\smash[t]{\overset{*}{\boldsymbol{q}_t}}||^2}{n} \geq \frac{2}{\alpha\bar{\gamma}_{2t}}\right)\\
				&+\mathbb{P}\Big(\sum_{r=0}^j[c_{qj}]_r^2 \sum_{r=0}^j\frac{||\boldsymbol{\Delta}_{qr}||^2}{n} \\
				& \qquad \qquad +  \sum_{r=0}^t[c_{qt}]_r^2\sum_{r=0}^t\frac{||\boldsymbol{\Delta}_{qt}||^2}{n} \geq\frac{\alpha\bar{\gamma}_{2t}\epsilon^2}{4}\Big)\\
				&\quad\leq (t+1)^2KK_t\exp\left\{-\frac{\kappa\kappa_tL\epsilon^2}{(t+1)^4(\log M)^{2t+2}}\right\},
			\end{aligned}
		\end{equation*}
		where the first and second term concentrates because $\smash[t]{\overset{*}{\boldsymbol{q}_t}}$ is Gaussian with $\mathbb{E}[[\smash[t]{\overset{*}{\boldsymbol{q}_t}}]_i^2]\geq\bar{\gamma}_{2t}^{-1}$ for $t\geq0$. We also recall from Lemma \ref{lemma:bound} that $\alpha\gamma_{2t}$ is lower bounded. For the last term, we first observe $\bar{\gamma}_{1t}=\mathbb{E}[([\boldsymbol{p}_t]_i)^2]=\sum_{r=0}^t\rho_{qr}([c_{qt}]_r)^2$; therefore, 
		\begin{equation}
			\sum_{r=0}^j[c_{qt}]_r^2\leq\frac{\bar{\gamma}_{2t}}{\min_{0\leq i\leq t} \{\rho_{qj} \}} \label{c_{qt} bound}
		\end{equation}
		is lower bounded. Thus, the last term concentrates due to $Q_j(a)$, $Q_t(a)$ and \eqref{c_{qt} bound}.
		
		$Q_t(c)$: From its definition,
		\begin{equation*}
			\boldsymbol{u}_{t+1}=\frac{1}{1-\bar{\alpha}_{2t}}\left[\frac{\gamma_w\boldsymbol{\omega}_q\boldsymbol{\xi}-\gamma_w\boldsymbol{\omega}_q^2\boldsymbol{q}_t}{\gamma_w\boldsymbol{\omega}_q^2+\bar{\gamma}_{2t}}\right]+\boldsymbol{q}_t:=g(\boldsymbol{q}_t,\boldsymbol{\omega}_q,\boldsymbol{\xi});
		\end{equation*}
		thus,
		\begin{equation*}
			\begin{aligned}
				&\left|\frac{g(\boldsymbol{q}_j,\boldsymbol{\omega}_q,\boldsymbol{\xi})^Tg(\boldsymbol{q}_t,\boldsymbol{\omega}_q,\boldsymbol{\xi})}{N}-\mathbb{E}\left[g(Q_j,W_q,\Xi)g(Q_t,W_q,\Xi)\right]\right|\\
				&\leq\left|\frac{g(\smash[t]{\overset{*}{\boldsymbol{q}}}_j,\boldsymbol{\omega}_q,\boldsymbol{\xi})^Tg(\smash[t]{\overset{*}{\boldsymbol{q}}}_t,\boldsymbol{\omega}_q,\boldsymbol{\xi})}{N}-\mathbb{E}\left[g(Q_j,W_q,\Xi)g(Q_t,W_q,\Xi)\right]\right|\\
				&+\frac{1}{N}\left|g(\boldsymbol{q}_j,\boldsymbol{\omega}_q,\boldsymbol{\xi})^Tg(\boldsymbol{q}_t,\boldsymbol{\omega}_q,\boldsymbol{\xi})-g(\smash[t]{\overset{*}{\boldsymbol{q}}}_t,\boldsymbol{\omega}_q,\boldsymbol{\xi})^Tg(\smash[t]{\overset{*}{\boldsymbol{q}}}_j,\boldsymbol{\omega}_q,\boldsymbol{\xi})\right|.
			\end{aligned}
		\end{equation*}
		Denote the right side as $T_1$ and $T_2$. For $T_1$, as $\omega_q=\sqrt{\frac{1}{\alpha}}s$ with probability $\alpha$ and $0$ with probability $1-\alpha$,
		\begin{equation}
			\begin{aligned}
				&\mathbb{E}\left[ e^{\lambda(g(Q_j,W_q,\Xi)g(Q_t,W_q,\Xi)-\mathbb{E}[g(Q_j,W_q,\Xi)g(Q_t,W_q,\Xi)])}|s \right] \\
				&=(1-\alpha)\mathbb{E}\left[e^{\lambda Q_jQ_t} \right]\\
				&+\alpha\mathbb{E}\left[\exp\left\{\lambda\left(\frac{\sqrt{\alpha}}{1-\bar{\alpha}_{2j}}\left[\frac{\gamma_ws\Xi-\bar{\gamma}_{w}s^2(Q_j/\sqrt{\alpha})}{\gamma_ws^2+\alpha\bar{\gamma}_{2j}}\right]+Q_j\right) \right.\right.\\
				& \qquad \left.\left. \times \left(\frac{\sqrt{\alpha}}{1-\bar{\alpha}_{2t}}\left[\frac{\gamma_ws\Xi-\bar{\gamma}_{w}s^2(Q_t/\sqrt{\alpha})}{\gamma_ws^2+\alpha\bar{\gamma}_{2t}}\right]+Q_t\right)\right\}|s\right]\\
				&\leq\exp\left\{\alpha\lambda^2\times\text{const}\right\},\label{eq:subexponential}
			\end{aligned}
		\end{equation}
		is sub-exponential with a factor $\alpha$. We use the fact that the conditional distribution of the exponent part is Gaussian with finite variance, as $\frac{\alpha}{1-\bar{\alpha}_{2j}},\bar{\gamma}_{2j}$ are upper bounded (Lemma \ref{lemma:bound}), and $s$ is strictly positive. Therefore, we have
		\begin{equation*}
			\mathbb{P}\left(T_1\geq\epsilon \right)\leq K\exp\left\{-\kappa n\epsilon^2\right\}.
		\end{equation*}
		according to Lemma \ref{concentration:sub-exp}. For $T_2$,
		\begin{equation*}
			\begin{aligned}
				T_2&\leq\frac{1}{N}\left|(g(\boldsymbol{q}_j,\boldsymbol{\omega}_q,\boldsymbol{\xi})-g(\smash[t]{\overset{*}{\boldsymbol{q}}}_j,\boldsymbol{\omega}_q,\boldsymbol{\xi}))^Tg(\boldsymbol{q}_t,\boldsymbol{\omega}_q,\boldsymbol{\xi})\right|\\
				&+\frac{1}{N} \left|g(\smash[t]{\overset{*}{\boldsymbol{q}}}_j,\boldsymbol{\omega}_q,\boldsymbol{\xi})^T\left(g(\boldsymbol{q}_t,\boldsymbol{\omega}_q,\boldsymbol{\xi})
				-g(\smash[t]{\overset{*}{\boldsymbol{q}}}_t,\boldsymbol{\omega}_q,\boldsymbol{\xi}\right)^T\right|.
			\end{aligned}
		\end{equation*}
		We consider the first term denoted as $T_2'$, while the second term can be considered similarly.
		\begin{equation*}
			\begin{aligned}
				&\mathbb{P}\left((T_2')^2\geq\epsilon^2\right)\leq\mathbb{P}\left(\frac{1}{N}\left\|g(\boldsymbol{q}_t,\boldsymbol{\omega}_q,\boldsymbol{\xi})\right\|^2\geq C \right) \\
				&+\mathbb{P}\left(\frac{1}{N}\left\|\left(\frac{1}{1-\bar{\alpha}_{2j}}\left[\frac{-\gamma_w\boldsymbol{\omega}_q^2}{\gamma_w\boldsymbol{\omega}_q^2+\bar{\gamma}_{2j}}\right]+1\right)\tilde{\boldsymbol{\Delta}}_{qj}\right\|^2\geq\frac{\epsilon^2}{2C} \right),
			\end{aligned}
		\end{equation*}
		where $\tilde{\boldsymbol{\Delta}}_{qj}:=\sum_{r=0}^j[c_{qj}]_r^2\boldsymbol{\Delta}_{qr}$.
	 	We note that exactly $n$ elements out of $\boldsymbol{\omega}_q$ are non-zero. Thus, the second term is upper bounded by $(t+1)^2KK_t\exp\left\{-\frac{\kappa\kappa_tL\epsilon^2}{(t+1)^4(\log M)^{2t+2}}\right\}$ similarly to $Q_t(b)$ because 
		$ \Big(\frac{\alpha}{1-\bar{\alpha}_{2t}}[\frac{-\gamma_wS_{min}^2}{\gamma_wS_{min}^2+\bar{\gamma}_{2t}}]+1\Big)^2$
		is upper bounded from Lemma \ref{lemma:bound}. 	The first term is then upper bounded by $K\exp\{-\kappa n\}+(t+1)^2KK_t\exp\left\{-\frac{\kappa\kappa_tL\epsilon^2}{(t+1)^4(\log M)^{2t+2}}\right\}$ because
		\begin{equation*}
			\begin{aligned}
				\frac{1}{N}||g(\boldsymbol{q}_t,\boldsymbol{\omega}_q,\boldsymbol{\xi})||^2&\leq\frac{2}{N}\Big(||g(\smash[t]{\overset{*}{\boldsymbol{q}}}_t,\boldsymbol{\omega}_q,\boldsymbol{\xi})||^2\\
				&+\left\|\left(\frac{1}{1-\bar{\alpha}_{2j}}\left[\frac{-\gamma_w\boldsymbol{\omega}_q^2}{\gamma_w\boldsymbol{\omega}_q^2+\bar{\gamma}_{2j}}\right]+1\right)\tilde{\boldsymbol{\Delta}}_{qj}\right\|^2\Big)
			\end{aligned}
		\end{equation*} concentrates according to \eqref{eq:subexponential} and $Q_t(b)$.
		
		$Q_t(d)$: This follows from $Q_t(c)$ along with the fact that $\mathbb{E}[Q_jU_{t+1}]=0$.
		
		$Q_t(e)$: This follows from \cite{rush2022finite}. 
	\end{proof}

	\subsection{Useful Lemmas}
	\begin{lemma}
		\label{concentration:max}
		Let $Z_1,\ldots,Z_N$ to be iid standard Gaussian random variables, then
		\begin{equation*}
			\mathbb{P}\left(\frac{1}{L}\sum_{\ell=1}^L\max_{j\in sec(\ell)}Z_j^2\geq3\log M\right)\leq\exp\left\{-\frac{L}{5}\log\frac{M}{70}\right\}.
		\end{equation*} 
	\end{lemma}
	This is Lemma 16 of \cite{rush2018error}.
	
	\begin{lemma}
		\label{concentration:chi}
		Let $Z_1,\ldots,Z_N$ to be iid standard Gaussian random variables and $0\leq\epsilon\leq1$, then
		\begin{equation*}
			\mathbb{P}\left(\left|\frac{1}{N}Z_i^2\right|\geq\epsilon\right)\leq2e^{-n\epsilon^2/8}.
		\end{equation*} 
	\end{lemma}
	This is Lemma 15 of \cite{rush2022finite}.
	
	\begin{lemma}
		\label{lemma:sum}
		For any scalars $a_1,\ldots,a_t$ and positive integer $m$, $(|a_1|+\ldots+|a_t|)^m\leq t(|a_1|^m+\ldots+|a_t|^m)$.
	\end{lemma}
	This is Lemma 16 of \cite{rush2022finite}.
	
	\begin{lemma}
		\label{concentration:sub-exp}
		Let $Z_1,\ldots,Z_N$ to be iid random variables that satisfy 
		\begin{equation}
			\mathbb{E}\left[e^{\lambda(Z_i-\mathbb{E}Z_i)}\right]\leq\exp\left\{\frac{\alpha\lambda^2}{\sigma^2}\right\},\quad\forall|\lambda|<\frac{1}{b},
		\end{equation}
		where $\alpha=\frac{n}{N}$, then we have
		\begin{equation*}
			\mathbb{P}\left(\frac{1}{N}\left|\sum_{i=1}^N(Z_i-\mathbb{E}Z_i)\right|\geq\epsilon \right)\leq\exp\left\{\frac{-n\epsilon^2}{2\sigma^2}\right\},\quad\forall0\leq\epsilon\leq\frac{\sigma^2}{\alpha b}.
		\end{equation*}
	\end{lemma}
	This is the concentration of sub-exponential random variables (here with a factor $\alpha$), see e.g.\ \cite{boucheron2013concentration}.
\end{document}